\documentclass{article}

\usepackage{arxiv}

\usepackage[utf8]{inputenc} % allow utf-8 input
\usepackage[T1]{fontenc}    % use 8-bit T1 fonts
\usepackage{hyperref}       % hyperlinks
\usepackage{url}            % simple URL typesetting
\usepackage{booktabs}       % professional-quality tables
\usepackage{amsfonts}       % blackboard math symbols
\usepackage{nicefrac}       % compact symbols for 1/2, etc.
\usepackage{microtype}      % microtypography
\usepackage{lipsum}
\usepackage{graphicx}
\graphicspath{ {plots/} }

\bibliography{biblio.bib}

\usepackage{alltt}
\usepackage{amsmath,amssymb,mathtools}
\usepackage{dsfont} % for \mathds{1}
\usepackage{amsthm}
\usepackage{bbm}
\usepackage{ifplatform} % to set graphics path conditionally on operating system
\usepackage{algorithm}
\usepackage{algpseudocode}
\makeatletter
\providecommand{\plist@algorithm}{\algorithmname~} % e.g., “Algorithm ”
\makeatother
\usepackage[title,titletoc]{appendix}
\usepackage{enumerate}
\usepackage{xparse}
\usepackage{float}
\usepackage{subcaption}
\usepackage{soul}
\usepackage{lmodern}

\makeatletter
\AtBeginDocument{%

  \@ifundefined{tagform@}{}{%
    \renewcommand{\tagform@}[1]{\maketag@@@{\ignorespaces#1\unskip\@@italiccorr}}
    \renewcommand{\eqref}[1]{(\ref{#1})}%
  }%
}
\makeatother

\theoremstyle{definition}
\newtheorem{definition}{Definition}[section]

\newtheorem{theorem}{Theorem}[section]
\newtheorem{lemma}[theorem]{Lemma}

\newtheorem{proposition}{Proposition}[section]
\newtheorem{result}{Result}[section]

\newcommand{\R}{\mathbb{R}}
\newcommand{\Z}{\mathbb{Z}}
\DeclareMathOperator*{\argmin}{\arg\!\min}
\DeclareMathOperator{\ran}{Ran}
\DeclareMathOperator{\dom}{Dom}
\newcommand{\cor}{\operatorname{Cor}}
\newcommand{\cov}{\operatorname{Cov}}
\newcommand{\var}{\operatorname{Var}}
\newcommand{\bbP}{\mathbb{P}}
\newcommand{\bbE}{\mathbb{E}}
\newcommand{\vect}[1]{\operatorname{vec} \paren*{#1}}
\newcommand{\vech}[1]{\operatorname{vech} \paren*{#1}}
\newcommand{\vechs}[1]{\operatorname{vechs} \paren*{#1}}
\newcommand{\bp}{\mathbf{P}}

\newcommand{\1}{\mathds{1}} % needs dsfont package
\DeclarePairedDelimiter{\paren}{\lparen}{\rparen}

\newcommand{\pr}[1]{\bbP\paren*{#1}}
\newcommand{\ex}[1]{\bbE\paren*{#1}}
\newcommand{\Rstudio}{\texttt{R}}
\newcommand{\lo}{\lambda_1}
\newcommand{\lt}{\lambda_2}
\newcommand{\ld}{\lambda_d}
\newcommand{\lx}{\lambda_X}
\newcommand{\ly}{\lambda_Y}

\newcommand{\by}{\mathbf{y}}
\newcommand{\byrs}{y_{1}^{(r)}, \ldots, y_{d}^{(r)}}
\newcommand{\byr}{\mathbf{y}^{(r)}}
\newcommand{\bys}{\mathbf{y}^{(1)}, \ldots, \mathbf{y}^{(n)}}
\newcommand{\byy}{\mathbf{Y}}
\newcommand{\btt}{\mathbf{t}}
\newcommand{\bt}{\boldsymbol{\theta}} 
\newcommand{\bg}{\boldsymbol{\gamma}} 
\newcommand{\bl}{\boldsymbol{\lambda}} 
\newcommand{\bz}{\boldsymbol{\zeta}}

\newcommand{\bps}{\boldsymbol{\psi}}
\newcommand{\br}{\boldsymbol{\varrho}}
\newcommand{\be}{\boldsymbol{\eta}}
\newcommand{\bo}{\boldsymbol{\omega}}

\newcommand{\ba}{\mathbf{b}}
\newcommand{\rst}[2]{\hat{\rho}_{#1}^{(#2)}}
\newcommand{\est}[2]{\hat{\eta}_{#1}^{(#2)}}
\newcommand{\zst}[2]{\hat{\zeta}_{#1}^{(#2)}}

\newcommand{\zer}{0, \ldots, 0}

\newcommand{\rss}{\boldsymbol{P}_{SS}}
\newcommand{\rij}{\boldsymbol{P}_{ij,ij}}
\newcommand{\rsij}{\boldsymbol{P}_{S,ij}}
\newcommand{\lrep}{\mathbf{L}} % in the reparametrization, P = LL'
\newcommand{\lreprow}{\mathbf{l}} % in the reparametrization, P = LL'; l_i is the row
\newcommand{\posq}[1]{e^{-2 \lambda_{{#1}}} \mathcal{I}_0 (2 \lambda_{{#1}})} % squared poisson pmf
\newcommand{\posqh}[1]{e^{-2 \hat{\lambda}_{{#1}}} \mathcal{I}_0 (2 \hat{\lambda}_{{#1}})} % squared poisson pmf with lambda hat
\newcommand{\bn}[1]{\boldsymbol{\nu}_{{#1}} } % vector of marginal parameters for whatever margins

\newcommand{\PhiDm}{\Phi_{d-1}}

\newcommand{\as}{A^{*} (\rs, \lx, \ly)}
\newcommand{\rs}{\rho^{*}}
\newcommand{\cs}{C_{\rs}}
\newcommand{\hs}{H^{*}}

\title{Novel Tau-Informed Initialization for Maximum Likelihood Estimation of Copulas with Discrete Margins}

\author{
 Anna van Es \\
  Research Institute for Statistics and Information Science, GSEM\\
  University of Geneva\\
  \texttt{anna.vanes@unige.ch} \\
   \And
 Eva Cantoni \\
  Research Institute for Statistics and Information Science, GSEM\\
  University of Geneva\\
  \texttt{eva.cantoni@unige.ch} \\
}

\begin{document}
\maketitle
\begin{abstract}
We study Gaussian-copula models with discrete margins, with primary emphasis on low-count (Poisson) data. Our goal is exact yet computationally efficient maximum likelihood (ML) estimation in regimes where many observations contain small counts, which imperils both identifiability and numerical stability. We develop three novel Kendall's tau-based approaches for initialization tailored to discrete margins in the low-count regime and embed it within an inference functions for margins (IFM) inspired start. We present three practical initializers (exact, low-intensity approximation, and a transformation-based approach) that substantially reduce the number of ML iterations and improve convergence. For the ML stage, we use an unconstrained reparameterization of the model's parameters using the log and spherical-Cholesky and compute exact rectangle probabilities. Analytical score functions are supplied throughout to stabilize Newton-type optimization. A simulation study across dimensions, dependence levels, and intensity regimes shows that the proposed initialization combined with exact ML achieves lower root-mean-squared error, lower bias and faster computation times than the alternative procedures. The methodology provides a pragmatic path to retain the statistical guarantees of ML (consistency, asymptotic normality, efficiency under correct specification) while remaining tractable for moderate- to high-dimensional discrete data. We conclude with guidance on initializer choice and discuss extensions to alternative correlation structures and different margins.
\end{abstract}

% keywords
\keywords{Kendall's tau \and Gaussian copula \and discrete margins \and Poisson \and unconstrained reparameterization \and optimization \and analytical gradients \and low-count regime \and starting values \and maximum likelihood}

\section{Introduction}

Understanding and modeling multivariate dependence are central tasks across statistics, econometrics, biostatistics, ecology and other sciences. Classical summaries such as (Pearson) correlation capture at best linear association under finite second moments and generally present with a number of fallacies \citep{nelsenIntroductionCopulas2006}. These limitations motivate shifting from moment-based summaries to distribution-free, rank-based descriptions of dependence. Copulas provide a principled way to separate the marginal behavior of each component from their joint dependence, delivering a flexible way of constructing multivariate distributions.

Consider a random vector $\byy=\left(Y_1, \ldots, Y_d\right)$ with joint distribution $H$, marginal distributions $F_{Y_j}, j = 1, \ldots, d$ and marginal parameter vector(s) $\bn{j}$ with $\bn{} = (\bn{1}, \ldots, \bn{d})$. A function $C:[0,1]^d \rightarrow[0,1]$ is called a copula if it is the distribution function of a random vector in $\mathbb{R}^d$ with $U(0,1)$ margins. By Sklar's theorem, there exists a copula $C$ with parameter vector $\bt$, such that, for any $\by=\left(y_1, \ldots, y_d\right) \in \mathbb{R}^d$, $H(\by ; \bt, \bn{}) = \pr{Y_1 \leq y_1, \ldots, Y_d \leq y_d  } =  C_{\bt}(F_{Y_1}(y_1), \ldots, F_{Y_d}(y_d)  )$ \citep{joeDependenceModelingCopulas2015,nelsenIntroductionCopulas2006}. If the margins $F_{Y_1}, \ldots, F_{Y_d}$ are continous, the associated copula is unique. If (some) margins are discrete, the copula is generally only unique on the product of the ranges but the decomposition still exists and may be used for modeling and inference. For booklength introduction to copulas, the reader is advised to consult \citet{nelsenIntroductionCopulas2006,hofertElementsCopulaModeling2018,joeDependenceModelingCopulas2015}. 

Many applications, like counts, ratings, zero-inflated measurements, among others, are inherently discrete, and linear correlation is ill-suited and can be misleading. When some - or all - margins are discrete, additional care is required. The lack of uniqueness of $C$ translates into likelihoods built from $2^d$ differences of inclusion-exclusion (rectangle) probabilities rather than densities, and rank ties complicate empirical procedures. Nevertheless, copula models remain attractive \citep{joeDependenceModelingCopulas2015,hofertElementsCopulaModeling2018}.

In this contribution, we work throughout with discrete (count) margins. In this setting, the joint probability mass function (pmf) is expressed via rectangle probability differences as follows \citep{joeDependenceModelingCopulas2015}:

\begin{equation}
    \label{eq:rect-prob}
    \begin{aligned}
        h(\by   ; \bt, \bn{}) & =\pr{Y_1=y_1, \ldots, Y_d=y_d  } = \\
        & =\pr {y_1-1<Y_1 \leq y_1, \ldots, y_d-1<Y_d \leq y_d  } = \\
        &= \sum_{t_1=0}^{1} \cdots \sum_{t_d=0}^{1} (-1)^{t_1 + \ldots + t_d}  C_{\bt}(u_{1_{t_1}}, \ldots, u_{d_{t_d}}  ) = \\
        &= \sum_{t_1=0}^{1} \cdots \sum_{t_d=0}^{1} (-1)^{t_1 + \ldots + t_d}  H \left( \left( y_1 - t_1, \ldots, y_d - t_d \right)   ; \bt , \bn{} \right),
        %& := \Delta_{Y_1} \ldots \Delta_{Y_d} H(\by   ; \bt, \bn{}),
    \end{aligned}
\end{equation}
with $u_{j 0} = F_{Y_j}(y_j), u_{j 1} = F_{Y_j}(y_j -1)$. If we consider a sample $\byr = \paren{\byrs}, r = 1,\ldots, n$ from $\byy$, then the associated log-likelihood is

\begin{equation}
    \label{eq:ll}
    \begin{aligned}
        \ell(\bt, \bn{}; \bys  ) & = \sum_{r=1}^n \log \left( h\left(\byr ; \bt, \bn{}\right) \right).
    \end{aligned}
\end{equation}

% Problem
A central difficulty with copulas for discrete margins is parameter estimation. When margins are discontinuous, the inverses of the cdf have plateaus, which can lead to biased estimates. If the margins are continuous, the functions $H(F_{Y_1}^{-1}(u_1), \ldots, F_{Y_d}^{-1}(u_d))$ and $H(F_{Y_1}^{-1}(u_{1\leftarrow}), \ldots, F_{Y_d}^{-1}(u_{d\leftarrow}))$ are the same, where $u_{j\leftarrow}$ indicates the limit of $u_j$ as it approaches from above \citep{trivediNoteIdentificationBivariate2017}. But if the margins are discontinous, then transformations that lead to a unique copula in the continous case now lead to different objects. Hence it is important to carefully consider the choice of reliable starting values for the copula parameter(s) when maximizing the log-likelihood in \autoref{eq:ll}. Typically in cases where the support is limited (e.g. Poisson with small mean parameters), the sample Pearson correlation of the counts is biased, so rank-based dependence measures such as Kendall's tau should be preferred. Our contribution aims at developing an exact yet computationally efficient maximum-likelihood (ML) procedure that leverages a Kendall's tau-based initializer of the parameters to produce stable starting values, scalable to moderate- and high-dimensional settings. We first provide a literature review of the state-of-the-art in the field.

% What others have done

A natural first idea could be to find a way to deconstruct the discreteness. If one somehow could make the discrete counts look continuous, the hurdles associated with discrete observations would disappear. This idea underlies the continous extension (CE) method (also called jittering), was first mentioned by \citet{denuitConstraintsConcordanceMeasures2005} and introduced by \citet{madsenMaximumLikelihoodEstimation2009}. It has since been reused in various forms, e.g., in \citet{shiLongitudinalModelingInsurance2014}. The idea is to add a random uniform jitter to the observed discrete data (counts) that would make it continous. Unfortunately, although seemingly attractive and simple, CE presents with significant drawbacks. In practice, reliable inference requires averaging over many jitters, which entails the choice of the number of jitters. Even when the number of jitters is high, the performance deteriorates markedly as the dimension grows or the latent dependence strengthens. More fundamentally, several authors have documented non-negligible bias and inefficiency of CE for discrete margins; see, in particular, \citet{nikoloulopoulosEstimationNormalCopula2013}. In short, CE is intuitive, but does not fix the problem. 

Another axis of parameter estimation discussion is the estimation method. There are two broad classes of frequentist estimators that do not rely on jittering. The first is the full maximum likelihood (ML) estimator. The second is the two-step inference functions for margins (IFM) estimator \citep{joeEstimationMethodInference1996}. In the absolutely continuous case with copula density $c$, the joint log-likelihood separates cleanly into a copula term and a sum of marginal terms: 

$$
\ell \left( \bt, \bn{} ; \bys \right) = \sum_{r=1}^{n} \left\{ \log c \left( F_{Y_1} \left( y_{1}^{(r)}\right), \ldots, F_{Y_d} \left( y_{d}^{(r)}\right) \right) + \sum_{j=1}^{d} \log \left( f_{Y_j} \left( y_{j}^{(r)}\right)\right) \right\} .
$$

This decomposition makes IFM an attractive option to optimize for the parameters separately: first estimate the marginal parameters (often under working independence), then maximize the copula part holding those margins fixed. However, when margins are discrete, this tidy decomposition often breaks down. The joint probability mass at $ \by = (\bys)^{\top}$ is no longer expressible via a copula density. As stipulated in \autoref{eq:rect-prob}, it is a $d$-fold finite difference of the copula cdf over a rectangle instead. For elliptical copulas - Gaussian copula in particular - this entails high-dimensional multivariate Normal (MVN) rectangle probabilities. Hence, classical ML becomes computationally heavy and the usual IFM logic ceases to be exact \citep{trivediCopulaModelingIntroduction2006}. One may still use IFM as a fast approximation, but it is expected to be statistically suboptimal.

Because of these obstacles, a substantial literature has explored alternatives. \citet{pittEfficientBayesianInference2006} showed that, for Gaussian copulas with arbitrary mixtures of discrete and continuous margins, full ML is often difficult in practice. Instead, they developed a Bayesian approach that estimates a Gaussian copula model. \citet{smithEstimationCopulaModels2011} echo the aforementioned difficulty and show that standard optimization algorithms for copula likelihood "are not doable for mixed (e.g., discrete and continuous) data". They propose a Bayesian latent variable method (augmenting discrete observations with latent continuous variables) to estimate copula models with discrete margins. By working in a Bayesian framework with data augmentation, they avoid the computationally heavy integrals of the full likelihood and improve inference when using Archimedean or Gaussian copulas for count data. Other Bayesian approaches on estimation of copula-based models with discrete margins are well summarized in \citet{smithBayesianApproachesCopula2013}. 

On the likelihood side, \citet{panagiotelisPairCopulaConstructions2012} analyze direct computation from cdf differences for discrete responses and highlight the fact that computational burden grows with dimension. \citet{louzadaModifiedInferenceFunction2016} put forward a modified IFM (MIFM) that uses data augmentation to better handle discreteness within a frequentist pipeline.

A comprehensive survey is given in \citet{nikoloulopoulosCopulaBasedModelsMultivariate2013}, covering simulated likelihood (SL), composite likelihood (CL), and Bayesian strategies. Another paper of \citet{nikoloulopoulosEstimationNormalCopula2013} examines CE and then develops efficient SL for the multivariate normal (MVN) copula with discrete margins, leveraging efficient MVN probability algorithms. 

In higher dimensions, pairwise or blockwise composite likelihood has been popular. The idea is to build an objective function from univariate and bivariate contributions \citep{zhaoCompositeLikelihoodEstimation2005}. CL can be attractive computationally, but by construction it ignores part of the dependence structure when estimating the margins, which can translate into noticeable efficiency loss relative to full ML. To mitigate this, \citet{nikoloulopoulosEfficientFeasibleInference2023} propose a weighted CL (WCL) that rebalances the pairwise components to improve efficiency. Its practical implementations are available (e.g., the \texttt{weightedCL} package in $\Rstudio$ \citep{nikoloulopoulosWeightedCLEfficientFeasible2022}). 

Very recently, \citet{kojadinovicCopulalikeInferenceDiscrete2024c} introduced a three-step "copula-like" procedure. They propose to first estimate the margins, then fit a copula density, and finally reconcile the two by an iterative proportional fitting procedure (IPFP). The idea stems back to \citet{geenensCopulaModelingDiscrete2020}, who proposed a different view on copula construction. The resulting pseudo-likelihood estimator enjoys attractive properties, but suffers from a common pseudo-likelihood drawbacks. Namely, one should expect efficiency below that of full ML under correct specification (see also the general quasi-likelihood theory of \citet{whiteMaximumLikelihoodEstimation1982}). Furthermore, to our knowledge, IPFP has not been extensively studied in the multivariate case yet.

% Gaussian copula in particular
Having provided arguments for the use of full maximum likelihood (ML) in the low-count regime, we now specialize to the Gaussian (multivariate normal, MVN) copula. The MVN copula inherits the correlation-driven dependence structure of the MVN distribution, allowing a wide spectrum of positive and negative associations and serving as a robust workhorse in a variety of applications. The price of this flexibility, in the discrete case, is that the mass at an observed count vector is a rectangle probability (cf. \autoref{eq:rect-prob}). Because the MVN cumulative distribution function (cdf) lacks a closed form, evaluating that mass entails repeated multidimensional integration. The expression for rectangle probabilities as in \autoref{eq:rect-prob} for a Gaussian copula with discrete margins is as follows:

\begin{equation}
    \begin{aligned}
    \label{eq:h-gaus}
    h_{ }(\by ; \bp, \bn{}) & = \int_{\Phi^{-1}\left[F_{Y_1}\left(y_1-1 \right)\right]}^{\Phi^{-1}\left[F_{Y_1}\left(y_1 \right)\right]} \ldots \int_{\Phi^{-1}\left[F_{Y_d}\left(y_d-1 \right)\right]}^{\Phi^{-1}\left[F_{Y_d}\left(y_d \right)\right]} \varphi_{\bp}\left(z_1, \ldots, z_d\right) d z_1 \ldots d z_d, \\
\end{aligned}
\end{equation}

\begin{equation}
    \label{eq:H-gaus}
    \begin{aligned}
H_{ }(\by ;  \bp, \bn{}) & = \int_{-\infty}^{\Phi^{-1}\left[F_{Y_1}\left(y_1 \right)\right]} \ldots \int_{-\infty}^{\Phi^{-1}\left[F_{Y_d}\left(y_d \right)\right]} \varphi_{\bp}\left(z_1, \ldots, z_d\right) d z_1 \ldots d z_d,
\end{aligned}
\end{equation}

where $\varphi_{\bp}$ denotes the standard centered $d$-variate normal density with correlation matrix $\bp$, chosen as correlation matrix as opposed to covariance matrix for identifiability \citep{alqawbaZeroinflatedCountTime2019}. Note that one could also include a set of covariates in this equation to handle regression settings.

A crucial issue is identification in a situation with low counts. In a thoughtful note, \citet{trivediNoteIdentificationBivariate2017} analyze the bivariate case under full ML and show how weak empirical support (few observed support points) can reinforce identification issues and bias dependence estimates. They mention identification improves as the observed support grows closer to the theoretical support. This observation is practically important for applications where many margins are concentrated near zero.

% Solution
We briefly summarize our contribution. Because of the applications we have in mind, we explicitly target the low-count regime and pursue full maximum likelihood (ML) for estimation. The motivation is twofold. First, under correct specification, ML retains its classical statistical guarantees, i.e. consistency, asymptotic normality, and efficiency (e.g., \citealp[Chapter 7]{CasellaStatisticalInference}), whereas continuous extension (CE) and pseudo/composite surrogates generally do not. Second, with careful numerical design the discrete ML is tractable even when the likelihood involves rectangle probabilities. Concretely, we compute stable starting values via an inference functions for margins (IFM)-type step, which typically reduces the number of iterations and improves convergence behaviour. We use a method of moments estimator for the marginal parameters and an estimator based on Kendall's tau for the copula parameters.

Within this framework, our substantive focus is on Kendall's tau for Gaussian copulas with count margins in the low-intensity regime (small means). The classical inversion idea - estimating the copula parameter via a target $\tau$ (see, e.g., \citet{mcneilQuantitativeRiskManagement2005}) is appealing but delicate once margins are non-continuous. We therefore review both foundational and recent results on Kendall's tau for copulas with discrete (or mixed) margins, emphasizing scenarios with many small counts. These low-count settings have been comparatively underexplored and exhibit particularities in identification and efficiency that are easy to overlook. In our estimation pipeline we leverage the theoretical properties of Kendall's tau alongside its empirical counterpart to gain efficiency and computational speed, while keeping full ML as the ultimate objective function. Section~\ref{sec:co-comput} presents three alternatives for the starting values using (an approximation of) Kendall's tau. 

We then adopt in Section~\ref{sec:opt-start} an unconstrained reparametrization for both margins and dependence (see Section~\ref{sec:reparam}), evaluate exact rectangle probabilities (rather than relying on jitters), and supply analytical gradients throughout to stabilize and accelerate the optimization. This strategy preserves the statistical benefits of ML while addressing the computational bottlenecks that arise with discrete margins - especially in the scenarios when dependence is strong or the dimension grows. 

This combination strikes a pragmatic balance: it retains the efficiency guarantees of ML under correct specification while controlling the computational burden inherent to Gaussian-copula models with discrete margins. Finally, because low-count data feature many zeros, a substantial fraction of the rectangle probabilities simplify, often collapsing to the leading term in the inclusion-exclusion sum, yielding further speed-ups without sacrificing exactness. Computational cost scales with the evaluation of multivariate normal (MVN) rectangle probabilities and with the number of nonzero inclusion-exclusion terms, which, in low-count data, depends primarily on the number of positive counts in each observation, rather than on dimension $d$ alone. This makes estimation in moderate to high dimensions possible, especially when zeros are prevalent. Section~\ref{sec:sim} goes into detail of high-accuracy multivariate normal (MVN) probability routines in $\Rstudio$ and presents a simulation study for different scenarios.

\section{Kendall's tau for a copula with discrete margins}

\label{sec:co-comput}

There are many ways to define dependence measures of two random variables. We adopt the concordance-discordance viewpoint: a pair is concordant when a larger value of one component tends to co-occur with a larger value of the other, and discordant otherwise. For this section, we will use the notation $(X,Y)$ for a bivariate pair of random variables. Let $(X,Y)$ be integer-valued discrete random variables whose joint distribution is $H$, with marginal cumulative distribution functions $F_X,F_Y$ marginal probability mass functions $f_X, f_Y$, $\bn{X}, \bn{Y}$ the (vectors) of parameters for each margin with $\bn{} = (\bn{X}, \bn{Y})$ and copula $C_{\theta}$. Formally, for two independent copies $(X_1, Y_1), (X_2, Y_2)$ of $(X,Y)$ with joint distribution function $H$, Kendall's tau is the difference between the probabilities of concordance and discordance, that is $\tau = P\left((X_1-X_2) (Y_1 - Y_2) >0\right) - P\left((X_1-X_2) (Y_1 - Y_2) <0\right)$ \citep{kendallNewMeasureRank1938}. For continuous margins this reduces to $\tau = 2P\left((X_1-X_2)(Y_1 - Y_2) >0\right) - 1$.

As noted by \citet{nelsenIntroductionCopulas2006}, in the continous case copulas are invariant under almost surely strictly increasing marginal transformations. Hence, the dependence structure (e.g. Kendall's tau) depends only on the copula and not on the marginal distributions. In one-parameter copula families (e.g. Gaussian, Clayton, Gumbel, Frank), there is a one-to-one mapping between $\tau$ and the copula parameter, enabling method-of-moments estimation by inversion of this relationship.

When margins are non-continuous (e.g. ordinal or count data), ties complicate both the definition and attainable range of Kendall's tau. In the continous setting ties occur with probability zero and can be ignored, but for discrete or mixed discrete-continous margins they must be accounted for. In particular, \citet{genestPrimerCopulasCount2007} emphasize that with discrete margins the copula alone does not define the dependence, and that the ranges of rank-based measures of association depend on the marginals. They also note that, for continuous margins, the empirical estimator of Kendall's tau, $\tau_n$, computed by replacing the probabilities of concordance and discordance with their empirical counterparts, will satisfy the equality

$$
\tau_n (X,Y) = \tau_n (U,V),
$$

with $U=F_X(X), \quad V=F_Y(Y)$, $F_X$ and $F_Y$ being the marginal distributions, because the quantile functions $F_X^{-1}$ and $F_Y^{-1}$ are continuous and strictly monotone. Thus $\tau_n (X,Y)$ is an unbiased estimator for $\tau (C_{\theta})$, with $\theta$ being the parameter of copula $C$. This is not the case when the margins are discrete. 

Finally, for discrete margins the copula (as a function) is not unique. Sklar's theorem guarantees the existence of a copula for any joint distribution, but uniqueness holds only when margins are continuous. If $X$ or $Y$ are discrete, different copulas can induce the same joint distribution $H$. As will be seen further, this non-uniqueness can be consequential both in applications and in theoretical derivations. 

While our focus is on (pairwise) bivariate distributions, it is worth noting that concordance measures have been extended to multivariate and mixed ordinal-continuous settings. \citet{genestPrimerCopulasCount2007} provided a rigorous treatment of rank correlation measures for non-continuous random variables, including definitions of multivariate Kendall's tau and Spearman's rho that handle ties. \citet{mesfiouiConcordanceMeasuresMultivariate2010} further proposed multivariate extensions of Kendall's tau for ordinal data, and studied tests of independence in such cases. These works ensure that concepts like "overall concordance" can be defined beyond the bivariate case, albeit with increasing complexity. In practice, however, most applications of Kendall's tau in discrete contexts remain bivariate (or pairwise in a multivariate array), often due to the complexity of joint discrete dependence models. 

Collecting the above points, when the data presents with ties, the definition of Kendall's tau must be generalized. On a population level, the probabilities of concordance and discordance no longer sum to one because ties occur with a positive probability in at least one margin. A more general identity (see e.g. \citet{denuitConstraintsConcordanceMeasures2005}) for the population $\tau$ is $\tau = 2P(\text{concordant}) - 1 + P(X_1 = X_2 \text{ or } Y_1 = Y_2 )$. This reduces to the continous case when the probability of ties is zero. \citet{KendallTau1945} mentioned this and introduced Kendall's tau-b, defined as $\tau_b = \tau / \sqrt{\pr{X_1 \neq X_2} \pr{Y_1 \neq Y_2}}$ \citep{genestPrimerCopulasCount2007}. This was however insufficient. Building on this idea, \citet{nikoloulopoulosModelingMultivariateCount2009} derived a general expression for Kendall's tau. The population version of Kendall's tau for $X$ and $Y$ is given by

\begin{equation}
    \label{eq:nik-tau}
    \begin{aligned}
        \tau_{\theta, \bn{}}\left(X, Y\right)= & \sum_{x=0}^{\infty} \sum_{y=0}^{\infty} h\left(x, y; \theta, \bn{}\right)\left\{4 C_{\theta}\left(F_X\left(x-1\right), F_Y\left(y-1\right)\right) - h\left(x, y; \theta, \bn{}\right)\right\} \\
        & +\sum_{x=0}^{\infty}\left(f_X^2\left(x\right)+f_Y^2\left(x\right)\right)-1,
    \end{aligned}
\end{equation}

where
$$
\begin{aligned}
h\left(x, y; \theta, \bn{}\right)= & C_{\theta}\left(F_X\left(x\right), F_Y\left(y\right)\right)-C_{\theta}\left(F_X\left(x-1\right), F_Y\left(y\right)\right) \\
& -C_{\theta}\left(F_X\left(x\right), F_Y\left(y-1\right)\right)+C_{\theta}\left(F_X\left(x-1\right), F_Y\left(y-1\right)\right)
\end{aligned}
$$
is the joint pmf of $X$ and $Y$. 

\autoref{eq:nik-tau} decomposes concordance probability into contributions from each combination of discrete outcomes, accounting for the jumps of the cdf at each count. The additional marginal terms $\sum_x f_X(x)^2$ (and similarly for $Y$) arise from tie probabilities in the margins. This reinforces the idea that Kendall's $\tau$ in discrete cases depends not only on the copula but also on the marginal distributions. As tie probabilities grow (e.g. highly discrete or low-dispersion margins), these terms increase, pulling $\tau$ toward zero and limiting its attainable range.

Notice that \autoref{eq:nik-tau} naturally splits into a term involving both the copula and the margins and a copula-free part. Define

\begin{equation}
    \label{eq:ab}
    \begin{aligned}
        A\left(\theta, \bn{}\right) := & \sum_{x=0}^{\infty} \sum_{y=0}^{\infty} h\left(x, y; \theta, \bn{}\right)\left\{4 C_{\theta}\left(F_X\left(x-1\right), F_Y\left(y-1\right)\right) - h\left(x, y; \theta, \bn{}\right)\right\}, \\
        & B_X(\bn{X}) =: P(X_2 = X_1), \\
        & B_Y(\bn{Y}) =: P(Y_2 = Y_1).
    \end{aligned}
\end{equation}

Then

\begin{equation}
    \label{eq:nik-tau-ab}
    \begin{aligned}
        \tau_{\theta, \bn{}} (X, Y) &= A\left(\theta, \bn{}\right) + B_X(\bn{X}) + B_Y(\bn{Y}) - 1.
    \end{aligned}
\end{equation}

In continuous copula models, one routinely exploits the monotone relationship between Kendall's tau and the copula parameter $\theta$ (for a given family) to estimate or interpret $\theta$. Consider a Gaussian copula with correlation parameter $\theta = \rho$. It can be shown that $\tau = (2/ \pi) \arcsin (\rho)$ for $\rho \in [-1,1]$, hence $\rho = \sin \left((\pi/2) \tau \right)$. In discrete settings, there is usually no closed-form invertible relationship available because $\tau$ also depends on the marginal distributions. For fixed marginal parameters, one can in principle compute the map $\theta \mapsto \tau(\theta)$ for a given family. For instance, \citet{nikoloulopoulosModelingMultivariateCount2009} computed $\tau(\theta)$ for Frank, Clayton, and Gumbel copulas under various marginal Poisson means $\lambda$. They observed that for small values of $\lambda$, the $\tau(\theta)$ curves for different copulas can nearly coincide, which could make it difficult to distinguish copula families based on data (identifiability issue). As $\lambda$ increases, these curves separate and tend toward their continuous counterparts. In practical terms, attempting to match Kendall's tau while ignoring the margins will inevitably introduce bias in the estimator of $\theta$. A common workaround is to estimate the margins first and then solve for $\theta$ using the margin-aware $\tau (\theta)$ relationship \citep{mcneilQuantitativeRiskManagement2005}. This is a two-step inference analogous to Inference Functions for Margins (IFM): first fit the margins, then fit the copula. The IFM approach is indeed widely used for copula models with discrete data, because full maximum likelihood for both the parameters of the margins and the copula can be challenging \citep{joeEstimationMethodInference1996}. Indeed, as for example in the case of Poisson margins, the support is infinite, and the pmf in \autoref{eq:rect-prob} is in practice computed via a finite truncation at some value $K$, chosen to control the truncation error. Consequently, exact joint pmf entails summing over the full $2^d$ of truncated combinations to form the rectangle probability, which quickly becomes prohibitive as $d, n$ and $\lambda$ grow. \citet{joeDependenceModelingCopulas2015} showed that IFM is asymptotically efficient for copula models under certain conditions and correct specification, although the inference can be complicated due to the fact that multiple copulas may yield the same discrete likelihood. Furthermore, efficiency worsens as dependence strengthens and can be non-negligeable even at moderate $d$. As it is based on the theory of estimating functions of Godambe, its asymptotic covariance is the Godambe (sandwich) matrix, rather than the inverse Fisher information matrix \citep{joeAsymptoticEfficiencyTwostage2005}. Finally, for discrete data, and in the case of Gaussian copula, the IFM doesn't eliminate the computation of rectangle probabilities as in \autoref{eq:h-gaus}. \citet{joeAsymptoticEfficiencyTwostage2005,joeDependenceModelingCopulas2015} also mentioned that if the maximum likelihood of the full multivariate likelihood is feasible, an IFM estimator can be considered as a good starting point. This motivated our ultimate choice of strategy that we support in the next Subsection, where we also discuss other less interesting options.

\subsection{Empirical estimator of Kendall's tau for discrete data}
As to the empirical counterpart of the Kendall's tau, in our experience, many estimators produce biased estimates for discrete data. This is due to the non-zero probability of ties in one of the margins or both margins simultaneously. Hence, the chosen estimator of Kendall's tau must account for these ties. Typically, the authors make a difference between a mass at zeros and ties elsewhere. 

The recent work by \citet{perroneKendallTauEstimator2023} improved the estimator for Kendall's $\tau$ originally provided by \citet{pimentelKendallsTauSpearmans}. Their ${\tau}_A$ adjusts $\tau_H$ \citep{pimentelKendallsTauSpearmans} for the probability of ties in each margin, yielding much lower bias. It also works in the presence of heavy ties. Their estimator $\tau_A$ has a closed-form expression that combines multiple parts. First, it takes into account Kendall's $\tau$ among the positive counts only (i.e. conditional on both $X,Y>0$), weighted by the probability both counts are non-zero. Next, it accounts for additional terms for concordance involving zeros. The exact expression is given by the following definition.

\begin{definition}[Estimator of Kendall's tau for zero-inflated count data from \citep{perroneKendallTauEstimator2023}]
    \label{def:tau-perrone}
    Let $(X,Y)$ be a bivariate random vector with distribution $H$. Denote by $X_{10}$ a positive random variable distributed as $X$ given that $Y=0, X_{11}$ a positive random variable distributed as $X$ given that $Y>0$. Similarly, $Y_{01}$ is a positive random variable distributed as $Y$ given that $X=0$, and $Y_{11}$ a positive random variable distributed as $Y$ given that $X>0$. 
    
    In addition, define the following probabilities: $p_{00}=\mathbb{P}[X=0, Y=0], p_{01}=\mathbb{P}[X=0, Y>0]$, $p_{10}=\mathbb{P}[X>0, Y=0], p_{11}=\mathbb{P}[X>0, Y>0], p_1^*=\mathbb{P}\left[X_{10}>X_{11}\right], p_2^*=\mathbb{P}\left[Y_{01}>Y_{11}\right]$, the probabilities of ties within the margins as $p_1^{\dagger}=\mathbb{P}\left[X_{10}=X_{11}\right]$, and $p_2^{\dagger}=\mathbb{P}\left[Y_{01}=Y_{11}\right]$ and define $\tau^{+}$ as Kendall's $\tau$ of $(X^{+}, Y^{+})$, i.e., such that $(X^{+}, Y^{+}) \overset{d}{=} (X,Y) | (X>0,Y>0)$ (away from zeros). 
    
    Then Kendall's $\tau$ is given by the following relation

    \begin{equation}
        \label{eq:tau-perrone}
        \tau_A=p_{11}^2 \tau^{+}+2\left(p_{00} p_{11}-p_{01} p_{10}\right)+2 p_{11}\left[p_{10}\left(1-2 p_1^*-p_1^{\dagger}\right)+p_{01}\left(1-2 p_2^*-p_2^{\dagger}\right)\right] .
    \end{equation}

\end{definition}

Given a sample $\left( \left(x^{(1)}, y^{(1)} \right), \ldots, \left(x^{(n)}, y^{(n)} \right)\right)$ of $(X,Y)$, all the versions of Kendall's tau $(\tau_A, \tau_H, \tau_b)$ can be estimated by replacing probabilities with relative frequencies. 
The estimator $\hat{\tau}_b$ \citep{KendallTau1945} is defined as 

$$
\hat{\tau}_b := \frac{P-Q}{\sqrt{ \left(P+Q+T_X\right) \left(P+Q+T_Y\right)}},
$$

with $P$ $(Q)$ the number of concordant (discordant) pairs, respectively, and $T_X$ $(T_Y)$ the number of pairs tied only on the $X$ $(Y)$ variable, respectively. It is implemented in $\Rstudio$ function \texttt{cor} (with \texttt{method = "kendall"}). 

Crucially, \citet{perroneKendallTauEstimator2023} proved their estimator $\hat{\tau}_A$ is consistent and asymptotically normal, and derived its achievable bounds. They showed via simulation that $\hat{\tau}_A$ outperforms an estimator they adapted from $\hat{\tau}_H$. In scenarios with many ties, the standard tie-corrected $\hat{\tau}_b$ estimator was severely biased and unstable, to the point that the authors did not even recommend using it. Hence in our work, focusing on counts, we will use the estimator given by \citet{perroneKendallTauEstimator2023} to estimate Kendall's tau.

\subsection{Comparison of options for starting points}
For the rest of this paper, our focus will be on a Gaussian copula $C_{\rho}$ with correlation parameter $\rho$, and on Poisson margins. The particularity of Poisson margins is that the probability of univariate ties has a closed form expression. Namely, for $X \sim \operatorname{Po} (\lx), Y \sim \operatorname{Po} (\ly)$, and their two respective independent copies $X_i, X_j$ and $Y_i, Y_j$, we have that

\begin{equation}
    \label{eq:pois-ties}
    \begin{aligned}
        \pr{X_j = X_i} & = \posq{X}, \quad \pr{ Y_j = Y_i} & = \posq{Y},
    \end{aligned}
\end{equation}

where $\mathcal{I}_0$ is the modified Bessel function of the first kind of order $0$. The proof can be found in Appendix~\ref{app:pois-ties}. 

Below we introduce our proposal and compare it to two alternatives that we have explored, but have finally discarded.

% our solution

\textbf{Option 1 - our proposal}
\label{sec:option1}

Building upon the arguments developed above, we present an IFM-type approach allowing for efficient estimation of starting values for the correlation parameter $\rho$ of the Gaussian copula. First the Poisson parameters $\lx$ and $\ly$ of the marginal distributions are estimated. Then, by using the empirical Kendall's tau $\hat{\tau}_A$ and the analytical expression for Kendall's tau as in \autoref{eq:nik-tau-ab}, one can obtain the initial values (denoted by $^{(s)}$ for "start") of $\rst{}{s}$ by solving $\hat{\tau}_A = A\left(\rho, \hat{\lambda}_X, \hat{\lambda}_Y\right) + B_X(\hat{\lambda}_X) + B_Y(\hat{\lambda}_Y) - 1$ for $\rho$.

More precisely, by combining \autoref{eq:nik-tau-ab} and \autoref{eq:tau-perrone}, the starting values for the copula parameters are computed as

\begin{equation}
    \label{eq:tau-eq}
    \begin{aligned}
        \hat{\lambda}_X & := \bar{x}, \quad \hat{\lambda}_Y := \bar{y}, \\
        \hat{A}_{XY} & := \hat{\tau}_{A}  - B_X(\hat{\lambda}_X) - B_Y(\hat{\lambda}_Y) + 1, \\
        \rst{}{s} & := \argmin_{\rho} \left( \hat{A}_{XY} - A(\rho, \hat{\lambda}_X, \hat{\lambda}_Y) \right)^2.
    \end{aligned}
\end{equation}

The \autoref{eq:tau-eq} can be generalized to higher dimensions by computing the pairwise Kendall's tau, i.e. for $\byy = \left(Y_1, \ldots, Y_d\right)$, consider a pair $(i,j) \in \left\{ (i,j) : 1 \leq i < j \leq d\right\}$:

\begin{equation}
    \label{eq:tau-eq-alld}
    \begin{aligned}
        \hat{\lambda}_i & := \bar{y}_i, \quad \hat{\lambda}_j := \bar{y}_j, \\
        \hat{A}_{ij} & := \hat{\tau}_{A} (Y_i, Y_j) - B_i(\hat{\lambda}_i) - B_j(\hat{\lambda}_j) + 1, \\
        \rst{ij}{s} & := \argmin_{\rho} \left( \hat{A}_{ij} - A(\rho, \hat{\lambda}_i, \hat{\lambda}_j) \right)^2.
    \end{aligned}
\end{equation}

Essentially, \autoref{eq:tau-eq-alld} requires solving a non-linear equation for each pair of variables. The Poisson parameters are estimated by the sample means, which is computationally trivial. The same holds for the evaluation of the $B_j$ terms. However, the function $A(\rho, \lambda_i, \lambda_j)$ in \autoref{eq:ab} involves infinite sums over the support of the Poisson distributions, which must be truncated in practice. The double sum over $x$ and $y$ can be computationally intensive, especially for large $\lambda$ where the Poisson pmf has a wide support. To mitigate this, one can truncate the sums at a chosen quantile of the Poisson distribution (e.g., 0.999 quantile) to capture most of the mass while keeping computations feasible. Still, solving for $\rho$ requires numerical root-finding methods, which can be slow if many pairs need to be processed. In our case with low counts, the computations are fast.

\textbf{Option 2}

For this option, we adopt an alternative representation of Kendall's tau, using Skellam distribution, which is the distribution of a difference of two independent Poisson random variables (see \autoref{def:skellam} in Appendix~\ref{app:skellam}). 

\begin{proposition}
    \label{prop:tau-skellam}
    Let $(X, Y)$ be integer-valued bivariate random vector with marginal distributions $F_X \sim \operatorname{Po}\left(\lx\right), F_Y \sim \operatorname{Po}\left(\ly\right)$, and let $C_\rho$ be a copula inducing the joint law $H$ of $(X, Y)$ (unique only on $\ran\left(F_X\right) \times \ran\left(F_Y\right)$). Let $\left(X_1, Y_1\right)$ and $\left(X_2, Y_2\right)$ be i.i.d. copies of $(X, Y)$, and define

    $$
    D_X:=X_2-X_1, \quad D_Y:=Y_2-Y_1 .
    $$

    Then:

    \begin{enumerate}[(i)]
        \item \label{prop1} $D_X \sim \operatorname{Skellam}\left(\lx, \lx\right)$ and $D_Y \sim \operatorname{Skellam}\left(\ly, \ly\right)$.
        \item \label{prop2} There exists a (grid-unique) subcopula $\cs$ on $\ran\left(F_{D_X}\right) \times \ran\left(F_{D_Y}\right) \subset[0,1]^2$ such that the joint cdf $\hs$ of $(D_X, D_Y)$ satisfies

            $$
            \hs\left(d_x, d_y; \rs, \lx, \ly\right)=\cs\left(F_{D_X}\left(d_x\right), F_{D_Y}\left(d_y\right)\right).
            $$

        In general $\cs \neq C_\rho$ and the parameter $\rs$ associated with a chosen parametric representation of $\cs$ need not equal $\rho$.
        \item \label{prop3} Fix a one-parameter copula family $\{C_{\theta}(\cdot ): \theta \in(-1,1)\}$ (e.g., the Gaussian family) and define

        \begin{align*}
            \hs\left(d_x, d_y ; \theta, \lx, \ly \right) & := C_{\theta}\left(F_{D_X}\left(d_x\right), F_{D_Y}\left(d_y\right) \right), \\
            h^*(0,0 ; \theta, \lx, \ly) &:=\hs(0,0 ; \theta, \lx, \ly)-\hs(-1,0 ; \theta, \lx, \ly)-\hs(0,-1 ; \theta, \lx, \ly)+\hs(-1,-1 ; \theta, \lx, \ly) .
        \end{align*}

            Let $\posq{j}:=\pr{D_j=0}=e^{-2 \lambda_j} I_0\left(2 \lambda_j\right)$ for $j= \{X,Y\} $. Then there exists a unique $\rs \in (-1,1)$ such that the Kendall's $\tau$ of $(X, Y)$ defined in \autoref{eq:nik-tau-ab} satisfies

            $$
            \tau(X, Y)=4 \hs\left(-1,-1 ; \rs, \lx, \ly\right)-1+\posq{X}+\posq{Y}-h^*\left(0,0 ; \rs, \lx, \ly\right) .
            $$
    \end{enumerate}

    Moreover, if $\psi_{\text {Sk }}$ denotes the Pearson-copula link mapping the family parameter $\theta$ to $\cor\left(D_X, D_Y\right)$ under $\hs(\cdot, \cdot ; \theta)$, then:

    \begin{enumerate}[(i)]
        \setcounter{enumi}{3}
        \item \label{prop4} $\psi_{\text {Sk }}$ admits the Hermite-Price power series $\psi_{\text {Sk }}(\rho)=\sum_{k \geq 1} b_k \rho^k$ with $b_k= \left\{a_k\left(F_{D_X}\right) a_k\left(F_{D_Y}\right)\right\} /\left(k ! \sigma_X \sigma_Y\right)$ and $a_k(F)=\int_{-\infty}^{\infty}F^{-1} (\Phi (t)) H_k(t) \varphi(t) d t$, with $H_k(t)$ the probabilists' Hermite polynomial of degree $k$ and $\varphi$ the standard normal pdf. 
        \item $\psi_{\text {Sk }}$ is real-analytic on $(-1,1)$, continuous on $[-1,1]$, strictly increasing, $\psi_{\text {Sk }}(0)=0$.
        \item $\psi_{\text {Sk }}$ is not identity but it converges to identity when $\lx, \ly \to \infty$.
        \item $\cor (X,Y)$ is not equal to $\rho$, and $\cor (D_X, D_Y)$ is not equal to $\rs$.
        \item \label{prop8} $\rs=\psi_{\mathrm{Sk}}^{-1}(\cor(X, Y))$.
        \item \label{prop9} In general, $\rs \neq \rho$.
    \end{enumerate}
\end{proposition} 

Notice that (\ref{prop1}) and (\ref{prop3}) are specific to Poisson margins, while (\ref{prop2}) and (\ref{prop9}) hold for any discrete bivariate distribution. Points (\ref{prop4})-(\ref{prop8}) are specific to Poisson margins and their Skellam differences. We notice Skellam margins follow from independence of copies and closure of Poisson under convolution of differences. By invoking Sklar's theorem for non-continuous margins (Theorem $2.3.3$ in \citet{nelsenIntroductionCopulas2006}), one can obtain a subcopula representation unique on $\ran\left(F_{D_X}\right) \times \ran\left(F_{D_Y}\right)$. The univariate tie probabilities are $B_X (\lx):= \pr{D_X=0} = \posq{X}, B_Y(\ly):= \pr{D_Y=0} = \posq{Y} $ by using \autoref{def:skellam}. 

Here, we briefly discuss the outline of the proof of \autoref{prop:tau-skellam}. The full proof can be found in Appendix~\ref{app:tau-skellam} and goes through several intermediate Lemmas. Define $\as : = 4 \pr{D_X \leq -1, D_Y \leq -1 } - \pr{D_X=0, D_Y=0}$. Conditioning on $(X_1,Y_1)$ and using independence of $(X_1,Y_1)$ and $(X_2,Y_2)$, it follows $P\left(D_X \leq -1, D_Y \leq -1 \right) = \ex{H(X - 1, Y-1)}$, $\pr{D_X=0, D_Y=0} = \ex{h(X,Y)}$, so that $\as$ is a functional of the joint law $H$ (and thus of copula $C_{\rho}$). By Sklar's theorem with discrete margins and Lemma \ref{lem:conv-copula-discrete} in Appendix~ \ref{app:tau-skellam}, there exists a subcopula on $\ran (F_X) \times \ran (F_Y) \subset [0,1]^2$ such that $H^*(d_x,d_y; \rs, \lx, \ly) = \cs (F_{D_X}(d_x), F_{D_Y} (d_y))$. The map $\cs$ is grounded and two-increasing on its (grid) domain and satisfies the subcopula boundary conditions $\cs(u,1) = u$ for $u \in \ran (F_X)$, $\cs(1,v) =v$ for $v \in \ran (F_Y)$. 

However, the discretized copula is not unique. Passing from $(X,Y)$ to the difference pair $(D_X,D_Y)$ pushes the support to $\mathbb{Z}^2$ and changes the margins to Skellam, retaining the same Poisson parameters $\lx, \ly$ at the marginal level. One can nevertheless show that there exists a correlation parameter $\rs$ (for a chosen copula family applied to $(D_X,D_Y)$) such that the Kendall's tau computed from the representation in \autoref{prop:tau-skellam} matches the original tie-corrected tau. Because the difference map is neither monotone nor linear on the observed (discrete) data (even not for the Gaussian copula), copula invariance does not apply and, in general, $\rho \neq \rs$. 

The illustration in Appendix~\ref{app:option2} shows that, while conceptually informative, \autoref{prop:tau-skellam} offers limited practical utility, at least for the computation of the starting values, and we therefore do not pursue our investigations in this direction.
 
% Kendall's tau approximation route
\textbf{Option 3}

For large parameters $\lambda$, the Poisson law is very well approximated by a $\mathcal{N}(\lambda, \lambda)$, as a result of the central limit theorem. A direct implication is that the tie probabilities shrink as $\lambda$ grows. Indeed, it is easy to see that the probability of univariate ties converges to zero when $\lambda$ tends to infinity. In the bivariate setting in practice, it translates to the fact that the probability of simultaneous ties in both margins decreases as support spreads. Consistent with this, \citet{nikoloulopoulosCopulaBasedModelsMultivariate2013} show that for Poisson means of around $10$, Kendall's tau is well approximated by its continuous counterpart. Similarly, \citet{trivediNoteIdentificationBivariate2017} note that as the mean increases, the empirical support of the Poisson distribution approaches its theoretical $(\R_+)$ limit. It is therefore natural to examine how these large-mean approximations translate into tie probabilities and rank-based dependence, in particular Kendall's tau. 
One practical route is to inspect the map $\rho \mapsto \tau(\rho; \lx, \ly)$ defined in \autoref{eq:nik-tau-ab} and to seek a simple, invertible approximation that links the copula parameter $\rho$ to Kendall's tau for parameter initialization and interpretation. 
Based on the analysis of Appendix~\ref{app:option3}, one could wonder whether the relationship could be well approximated by an adjusted arcsine function of the form appearing in \autoref{prop:tau-asin} below, whose proof is given in Appendix~\ref{app:prop-tau-asin}.

\begin{proposition}[Kendall's tau approximation in arcsine form]
    \label{prop:tau-asin}
    Let $X,Y$ be Poisson random variables with respective parameters $\lx, \ly$ and Gaussian copula with parameter $\rho$. Furthermore, let $(X_1,Y_1), (X_2,Y_2)$ be two independent copies from $(X,Y)$. Finally, let
    
    $$
    m(\rho, \lx, \ly) = a(\lx, \ly) \left( 2/ \pi \right) \arcsin \left(  b(\lx, \ly) \rho \right),
    $$ 

    with 
    
    \begin{enumerate}[(i)]
        \item $a(\lx, \ly)  \in (0,1]$ a scaling function with $\lim_{\lx, \ly \to \infty} a(\lx, \ly) = 1$, 
        \item $b(\lx, \ly)  \in (0,1]$ a shrinkage function adjusting for the correlation attenuation due to discreteness with $\lim_{\lx, \ly \to \infty} b(\lx, \ly) = 1$.

    \end{enumerate}

    Then Kendall's tau $\tau(X,Y)$ for a Gaussian copula with discrete margins and parameter $\rho$ can be approximated by $m(\rho, \lx, \ly)$.
\end{proposition}

Two simple possible candidates for $m(\rho, \lx, \ly)$ are:

\begin{itemize}
    \item Logistic shrinkage of $\rho$, inspired by what we observed by numerical optimization for values of $a,b$ over a grid of $\lambda$'s and $\rho$'s. Set $a(\lx, \ly) = 1$, and $b(\lx, \ly) = \left(1 + \exp(-(\lx + \ly)) \right)^{-1}$, which yields
    \begin{equation}
        \label{eq:tau-asin-logis}
        m_l(\rho, \lx, \ly) = \frac{2}{\pi} \arcsin \left(  \left(1 + \exp(-(\lx + \ly)) \right)^{-1} \rho \right).
    \end{equation}
    This is easy to invert but too crude for small $\lambda$ (see below).
    \item $\tau_b$ inspired scaling. Set $a(\lx, \ly) = \sqrt{\paren{1-\posq{X}} \paren{1-\posq{Y}}}$, and $b(\lx, \ly) = 1$. The expression is adapted from $\tau_b$ in \citet{genestPrimerCopulasCount2007}. In their original formula, the coefficient is rather $a(\lx, \ly)^{-1}$, however, as $a(\lx, \ly)^{-1} \geq 1 \quad \forall \lx, \ly \in \R^+$, and would push $m(\rho, \lx, \ly)$ beyond the $[-1,1]$ range, we use $a(\lx, \ly)$ instead. This yields
    \begin{equation}
        \label{eq:tau-asin-taub}
        m_b(\rho, \lx, \ly) = \frac{2}{\pi} \sqrt{\paren{1-\posq{X}} \paren{1-\posq{Y}}} \arcsin \left(  \rho \right).
    \end{equation}
\end{itemize}

\begin{figure}[ht]
    \centering
    \includegraphics[width=\textwidth]{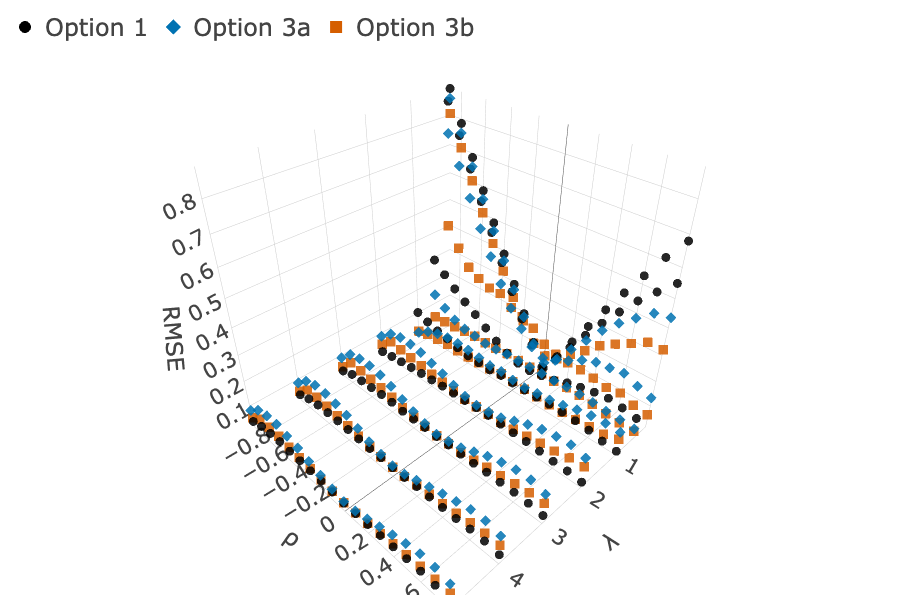}
    \caption{The relationship between the values of $\lambda, \rho$ and RMSE for the three different approaches.}
    \label{fig:approxTau3D}
\end{figure}

We explored further the suitability of these two approximations together with Option~1 numerically. We proceeded by setting $\lx = \ly$ for a grid of lambdas $\boldsymbol{\lambda} = \{0.05, 0.1, 0.5, 1, 2, 3, 4, 5 \}$ and $\boldsymbol{\rho} = \{- 0.9, -0.8, \ldots, 0.8, 0.9\}$. For each combination of parameters, we computed:

\begin{enumerate}
    \item \label{eq:option1} (Option 1) $\rst{}{s} = \argmin_{\rho} \left( \hat{\tau}_{A} - B_X(\hat{\lambda}_X) - B_Y(\hat{\lambda}_Y) + 1 - A(\rho, \hat{\lambda}_X, \hat{\lambda}_Y) \right)^2$ as in \autoref{eq:tau-eq}.
    \item \label{eq:option3a} (Option 3a) $\rst{}{l} = \left(1 + \exp(-(\hat{\lambda}_X + \hat{\lambda}_Y)) \right) \sin \left( \frac{\pi}{2} \hat{\tau}_A\right)$, derived from \autoref{eq:tau-asin-logis},
    \item \label{eq:option3b} (Option 3b) $\rst{}{b} = \sin \left( \frac{\pi}{2}  \left( \sqrt{\paren{1-\posqh{X}} \paren{1-\posqh{Y}}} \right)^{-1}  \hat{\tau}_A\right)$, derived from \autoref{eq:tau-asin-taub}.
\end{enumerate}

For each situation, we computed RMSEs. \autoref{fig:approxTau3D} shows plots of RMSEs as a function of $\lambda$'s and $\rho$, for each method. Average RMSE over all parameters' combinations rounded up to three decimals for each approach yielded $0.108$ for $\rst{}{s}$, $0.146$ for $\rst{}{l}$, and $0.147$ for $\rst{}{b}$.  It must be noted that the RMSEs for all methods decreased when the copula parameter was closer to $0$. The estimator $\rst{}{l}$ produced worst RMSEs for small values of $\lambda$ when $\rho$ were positive, however, it performed best with higher values of $\lambda$'s for all values of $\rho$. The estimator $\rst{}{b}$ yielded the worst results with higher values of $\lambda$'s, but the best results when $\lambda$'s were very small $(0.1,0.5)$ and $\rho$ either close to $0$ $(0.1,0.2)$ or $0.9$. The estimator $\rst{}{s}$ produced the most consistent results overall.

In summary, while the alternative viewpoint on Kendall's tau as in \autoref{prop:tau-skellam} presents interesting insights, its practical use is limited. Furthermore, although the idea of approximation is appealing, deriving a tractable approximation of Kendall's tau stays a nontrivial task. Despite the fact that the functions proposed above might work very well in some scenarios, we decided to opt for the method yielding the smallest average RMSE overall, namely, the method described in \autoref{eq:tau-eq}. What we propose is thus to use IFM-type approach for the computation of the initial (starting) parameters, ultimately used for the full maximum likelihood approach, leveraging its standard large-sample guarantees, i.e. consistency and asymptotic normality (e.g., \citet[Chapter 7]{CasellaStatisticalInference}).

\section{Parameters estimation}
\label{sec:opt-start}

As mentioned before, the computation of the MVN cdf probabilities is challenging. A number of efficient numerical methods are available and well documented (see \citet{nikoloulopoulosCopulaBasedModelsMultivariate2013} for a survey and software pointers in \Rstudio). A solid choice are the algorithms of \citet{genzMvtnormMultivariateNormal2000} and \citet{GenzBretzComputationMultivariateNormal2009}, as implemented in the \texttt{mvtnorm} package for $\Rstudio$ \citep{genzMvtnormMultivariateNormal2000}. These are randomized (quasi-)Monte Carlo methods with variance-reduction heuristics that handle arbitrary covariance structures with high accuracy and scale to fairly large dimensions (up to $1000$ according to package documentation). Their drawback is non-determinism and potential slowness when high precision is required. In low to moderate dimensions the deterministic Miwa algorithm \citep{miwaEvaluationGeneralNonCentred2003} can be used instead, which can be faster and reproducible when $d$ is small. 

Because direct ML hinges on these MVN integrals, several approximations of the log-likelihood have been proposed. The distributional transform (DT) treats the discrete margins as though they were continuous after a suitable transform, i.e. by stochastically "smoothing" the cdf at its jumps. The technique yields a tractable "continuous-data-like" likelihood \citep{KaziankaPilzCopulabasedGeostatisticalModeling}. While attractive computationally, DT is not universally valid: \citet{nikoloulopoulosEfficientEstimationHighdimensional2016} show that DT can be biased under heavy discretization and strong dependence. By contrast, the continuous extension (CE) or jittering approach \citep{madsenMaximumLikelihoodEstimation2009} is, in principle, exact up to Monte Carlo error if one averages over (infinitely) many jitters, but the computation is costly and, as discussed earlier, it suffers from efficiency and bias issues in challenging regimes. A recent nuance is provided by \citet{hughesOccasionalExactnessDistributional2021}, who identify model classes for which DT is (nearly) exact and propose a diagnostic to assess DT-exactness on a given dataset (e.g., certain autoregressive or mixed-effects structures with common count margins). Thus, DT could be justified in very specific cases, but it is not a general remedy for discrete Gaussian-copula likelihoods.

In line with \citet{nikoloulopoulosEfficientFeasibleInference2023}, fully specified multivariate normal (MVN) copula models attain the efficiency bound under correct specification. To make the optimization numerically robust and unconstrained, we propose to reparametrize both the marginal and dependence components. Then, we develop the expression for the analytical gradient to speed up the computations and improve convergence.

In the absence of covariates, one could use box constraints (e.g. \texttt{L-BFGS-B} method in \texttt{optim}), but in the case where the covariates are present, the problem turns into nonlinear constrained optimization, much more difficult to solve. Reparametrization solves this problem completely. Ensuring the model's parameters are in $\R^{d+d(d-1)/2}$ allows to obtain the following advantages. First, we ensure that every iteration step maps to a feasible vector of estimated parameters. This avoids objective function's evaluation producing NaNs or undefined derivatives and makes the process more smooth. Next, this allows for better behaviour of all the functions near the boundaries, e.g. when $\rho \to \pm 1, \lambda \to 0$. Standard texts, like \citet{NumericalOptimization2006}, emphasize how poor scaling or lack of curvature degrades second-order methods, and that scaling or variable changes (including reparameterizations) fix it. Thus, the optimization algorithm will be more likely to converge to a global maximum (when the objective function is a log-likelihood).

\subsection{Unconstrained model parameters}
\label{sec:reparam}

Let $\mathbf{Y} = (Y_1, \ldots, Y_d)$ have Poisson margins, coupled by a $d$-dimensional Gaussian copula with correlation matrix $\bp = \left(\rho_{ij}\right)_{1\leq i,j\leq d}$ with entries $\rho_{ii} = 1, \rho_{i j} = \rho_{j i} \in (-1,1)$. Consider the strict lower-triangular index set $\mathcal{K} := \left\{(i,j): 1 \le j < i \le d\right\}$ and define columnwise strict half-vectorization $\vechs{\bp}:= \boldsymbol{\varrho} = \left(\rho_{21}, \ldots, \rho_{d1}, \rho_{32}, \ldots, \rho_{d2}, \ldots, \rho_{d(d-1)}\right)$, i.e. stacking the columns of the lower-triangular entries of $\bp$ without the main diagonal. Let $p= |\vechs{\bp}| = d(d-1)/2$, the number of copula (correlation) parameters.

The original parameter vector is $\bg = \left( \lo, \ldots, \lambda_d, \vechs{\bp}\right)$, where $\lambda_j >0$. The map from $\bg$ to the unconstrained parameter vector $\boldsymbol{\psi} \in \R^{d+d(d-1)/2}$ is as follows.

For the marginal (Poisson) parameters, we use a log link: 

\begin{equation}
  \label{eq:lam-reparam-tau}
  \eta_j = \log(\lambda_j), \quad j=1,\ldots,d, \quad \eta_j \in \mathbb{R}.
\end{equation}

This transformation ensures that $\lambda_j$ is always positive. 

As to rhos, \citet{lucchettiSphericalParametrisationCorrelation2024} proposed a reparametrization of the copula parameters using the spherical coordinates. They have shown it can sometimes have numerical advantages. Using Cholesky decomposition, write $\bp=\lrep \lrep^{\top}$, with $\lrep$ a lower triangular matrix. Let the $i$-th row of $\lrep$ be defined as

\begin{equation}
    \label{eq:cholesky}
        \lreprow_i=\left( \alpha_{i 1}, \ldots, \alpha_{i i}, \zer \right),
\end{equation}

with $\alpha_{1 1} = 1$, $q_{i j}:=\prod_{k=1}^{j-1} \sin \omega_{i k}$ (with $q_{i 1}=1$) and with the vector $\boldsymbol{\alpha}_i \in \R^{i}$ given recursively by

\begin{equation}
    \label{eq:spherical-coordinates} \alpha_{i j} =
    \begin{cases}
        \cos \omega_{i 1}, \quad & j=1,\\ 
        \left(\prod_{k=1}^{j-1} \sin \omega_{i k}\right) \cos \omega_{i j} = q_{i j} \cos \omega_{i j}, \quad & 2 \leq j \leq i -1, \\ 
        \prod_{k=1}^{i-1} \sin \omega_{i k} = q_{i i}, \quad & i=j,
    \end{cases}
\end{equation}

where $\omega_{i j} \in[0 ; \pi]$. To remove the box constraints on $\omega_{i j}$, we follow the authors' suggestion and use the logistic map to obtain an unconstrained parameter $\zeta_{i j} \in \R$:

\begin{equation}
    \label{eq:reparam-copula}
    \zeta_{i j} = \log\left(\frac{\omega_{i j}}{\pi - \omega_{i j}}\right), \quad \omega_{i j} = \frac{\pi \exp(\zeta_{i j})}{1 + \exp(\zeta_{i j})} = \pi \Lambda (\zeta_{i j}).
\end{equation}

This yields a smooth bijection $\text{SCP}(\bz): \R^p  \to (0,1)^p; \bz \mapsto \lrep (\bz) \mapsto \bp(\bz)$, ensuring the angles $\zeta_{i j} \in \R$ for all $(i,j) \in \mathcal{K}$. We stack the angle parameters over the strict lower triangle, columnwise, in the same order as $\vechs{\bp}$:
$$
\bz : =  (\zeta_{21},\ldots,\zeta_{d1},\ \zeta_{32},\ldots,\zeta_{d2},\ \ldots,\ \zeta_{d(d-1)})^{\top},\qquad
\boldsymbol{\omega} := (\omega_{21},\ldots,\omega_{d1},\ \omega_{32},\ldots,\omega_{d2},\ \ldots,\ \omega_{d(d-1)})^{\top}.
$$

Finally, by combining Equations \ref{eq:lam-reparam-tau}-\ref{eq:reparam-copula}, the reparametrized parameter vector is

\begin{equation}
  \label{eq:psi-reparam}
  \boldsymbol{\psi} = (\be, \bz) \quad \in \R^{d+d(d-1)/2}, \quad \lambda_j = \exp(\eta_j), \quad j=1,\ldots,d, \bp = \text{SCP}(\bz).
\end{equation}

This reparametrization typically improves numerical stability, avoids boundary hits, and permits gradient-based methods to operate on the real line. Next, we develop closed-form derivatives with respect to unconstrained parameters $\boldsymbol{\psi}$ to further accelerate and stabilize the maximum likelihood optimization.

\subsection{Gradient of the log-likelihood}
\label{sec:gradient}

Let $\bg = (\boldsymbol{\lambda}, \br)$ denote the original parameters and $\bps = (\be, \bz)$ the unconstrained parameters introduced in Subsection~\ref{sec:reparam}, with $\eta_k = \log (\lambda_k)$ and $\bp = \text{SCP} (\bz), \br = \vechs{\bp}$ via the spherical-Cholesky map. For any observation $\mathbf{y} \in \mathbb{N}^d$, and any corner vector $\mathbf{t}\in \left\{ 0,1\right\}^d$, write 

$$
H\left(\mathbf{y}-\mathbf{t}; \bg\right) = \Phi_d\left(\Phi^{-1}\left(F_{Y_1}\left(y_1 - t_1\right)\right), \ldots, \Phi^{-1}\left(F_{Y_d}\left(y_d - t_d\right)\right); \bp\right)
$$

for the multivariate normal rectangle probability in the discrete likelihood as in \autoref{eq:H-gaus}, with the inclusion-exclusion over $\boldsymbol{t}$. To compute the gradient of the log-likelihood, first we make use of Theorem~$3.3.4$ in \citet{tongMultivariateNormalDistribution2012}. It states that the conditional distribution in a multivariate normal distribution is multivariate normal. Let $\mathbf{Z} \sim \mathcal{N}_d(\boldsymbol{0}, \bp)$ be a multivariate normal random vector. Define $S = \{1, \ldots, d\} \backslash\{i, j\}$, and partition correlation matrix $\bp$ as $\rss, \rij, \rsij$, where $\rsij$ has rows consisting of $S$ rows and $(i,j)$ columns. Fix $\ba = (b_i,b_j)$. Then the conditional law $\boldsymbol{Z}_{S} \mid Z_i=b_i, Z_j=b_j \sim \mathcal{N}\left(\boldsymbol{\mu}_{S \mid i j}, \boldsymbol{\Sigma}_{S \mid i j}\right)$, where

$$
\boldsymbol{\mu}_{S \mid i j}=\rsij \rij^{-1} , \quad \boldsymbol{\Sigma}_{S \mid i j}=\rss-\rsij \rij^{-1} \rsij^{\top},
$$

and $\rij^{-1} = \frac{1}{1 - \rho_{i j}^2} \begin{bmatrix}
1 & -\rho_{i j} \\
-\rho_{i j} & 1
\end{bmatrix}$. Note that we condition here on $\left\{ i, j\right\}$ to provide an example, but this result could be generalized by conditioning on any set $\mathcal{S} \subset \left\{1, \ldots, d\right\}$. 

To compute the gradient of the log-likelihood with respect to both dependence and marginal parameters, we use the results by \citet{tongMultivariateNormalDistribution2012,plackettReductionFormulaNormal1954,lucchettiSphericalParametrisationCorrelation2024}. For the correlation parameters, we first compute the score defined as

\begin{align*}
    \boldsymbol{s}_{P} \left( \bg; \bys\right) &:= \frac{\partial \ell \left( \bg ; \bys\right)}{\partial \vechs{\bp}}.
\end{align*}

Then, by the chain rule,

$$
\frac{\partial \ell \left( \bps ; \bys\right)}{\partial \boldsymbol{\zeta}_{ } }= \boldsymbol{s}_{P} \left( \bg; \bys\right) \left[\frac{\partial \vechs{\bp}}{\partial \boldsymbol{\omega}_{ } }\right]^{}  \left[\frac{\partial \boldsymbol{\omega}_{ } }{\partial \boldsymbol{\zeta}_{ } }\right]^{}.
$$

For the Poisson means, let

\begin{align*}
    \boldsymbol{s}_{\lambda_k} \left( \bg; \bys\right) &:= \frac{\partial \ell \left( \bg ; \bys\right)}{\partial \lambda_k}
\end{align*}

be the score function of the original log-likelihood with respect to $\lambda_k$. Then, by the chain rule, 

$$
\frac{\partial \ell \left( \bps ; \bys\right)}{\partial \eta_k} = \frac{\partial \ell \left( \bg ; \bys\right)}{\partial \lambda_k} \frac{\partial \lambda_k}{\partial \eta_k}= \boldsymbol{s}_{\lambda_k} \left( \bg; \bys\right) \lambda_k, 
$$

with $\lambda_k=\exp \left(\eta_k\right), \quad k=1, \ldots, d$. Finally, the gradient of the log-likelihood with respect to the unconstrained parameters $\bps$ is provided in the following proposition.

\begin{proposition}[Gradient of the log-likelihood of a Gaussian copula with Poisson margins.]
    \label{def:grad}
    The gradient of the log-likelihood with respect to unconstrained parameters is as follows:
    \newline
    
    \small
    \begin{align*}
        & \nabla_{\ell} (\eta_1, \ldots, \eta_d, \zeta_{ij}; \bys) = \\
        & = \begin{bmatrix}
            % l1
            \sum_{r=1}^n \frac{1}{h\left(\by^{(r)} ; \bps\right)} \left[  \boldsymbol{s}_{\lo}^{(r)} \left( \bg; \by^{(r)}\right) \exp \left(\eta_1\right) \right] \\
            \vdots \\
            % ld
            \sum_{r=1}^n \frac{1}{h\left(\by^{(r)} ; \bps\right)} \left[  \boldsymbol{s}_{\ld}^{(r)} \left( \bg; \by^{(r)}\right) \exp \left(\eta_d\right) \right] \\
            \newline \\
            % zeta
            \sum_{r=1}^n \frac{1}{h\left(\by^{(r)} ; \bps\right)}  \left[\pi \Lambda\left(\zeta_{ij}\right)\left\{1-\Lambda\left(\zeta_{ij}\right)\right\} \boldsymbol{s}_{\bp} \left( \bg; \by^{(r)}\right)^{\top} \left[M_p D_d^{+} \left(\boldsymbol{I}_{d^2} + \boldsymbol{K}_d\right) \left(\boldsymbol{I}_d \otimes \lrep\right) \left[ \mathbf{e}_i \otimes \mathbf{g}_{i j}\right] \right] . \right]  \\
    \end{bmatrix},
    \end{align*}

    \normalsize
    where $\bp = \lrep \lrep^{\top}$ with $\lrep$ lower triangular, $\boldsymbol{I}$ the identity matrix, $\boldsymbol{K}_d$ the $d^2 \times d^2$ commutation matrix, so that $\vect{\lrep^{\top}} =\boldsymbol{K}_d \vect{\lrep}$, $\otimes$ the Kronecker product, $D_d$ the $d^2 \times d(d+1)/2$ duplication matrix such that $D_d \vech{\bp} = \vect{\bp}$ and $M_p$ a $p \times d(d+1)/2$ selection matrix such that $\vechs{\bp} = M_p \vech{\bp}$ \citep{MatrixDifferentialCalculus}, $\mathbf{g}_{i j} \in \mathbb{R}^d$ a column vector, whose first $i$ entries are partial derivatives of $i$-th rows of $\boldsymbol{l}_i$ with respect to reparametrized angles $\omega_{i j}$ and zeros afterward, and $\mathbf{e}_i$ the $i$-th canonical vector in $\mathbb{R}^d$.
\end{proposition}

The full computations are available in Appendix \ref{app:grad}. It must be noted that the end results for both marginal and copula parameters yields the same inclusion-exclusion formula for $2^d$ differences of multivariate normal cdf, evaluated at their conditional means, similar to the expression of the log-likelihood. However, the integrals to compute will be of reduced dimension $d-2$ and $d-1$ with respect to correlation and marginal parameters, respectively. Hence, this will be significantly cheaper to evaluate than the likelihood terms. Furthermore, it must be noted that in the case where a mass at multivariate zero $(\zer)$ is observed, both the expression for the log-likelihood and the gradient will be simplified, as only the first term of the $2^d$ differences will survive. Indeed, as we have $F_{Y_j}(-1) = 0, j = 1, \ldots, d$, it follows that $h(\zer; \bps) = H(\zer; \bps)$. This property will be helpful when evaluating those terms numerically. Especially in the cases where the Poisson means are low, this will significantly speed up the computations. 

\section{Simulation}
\label{sec:sim}

\subsection{Setting}

To test whether the proposed improvements yield the expected results, we performed a series of simulations under the following settings. For dimension $d=2$, we set $\rho = \rho_{21}= -0.5$ and consider a setting where Poisson parameters are small $(\bl = (0.5,1))$ and moderate $(\bl = (2,3))$. For dimension $d=4$, we set $\br = (-0.42, -0.23,  0.73,  0.21, -0.64,  0.18)$ and consider a mixed setting $\bl = (0.6, 2, 4, 0.8)$. We choose an unstructured correlation matrix to allow for a large mix of different correlations that would correspond to a large panel of real-life situations. The two different scenarios for the starting values are outlined in the next section. Furthermore, we compare the results with the analytical gradient that we derived and with the gradient computed numerically. We take a sample size of $n=500$ across all scenarios and do $N=100$ repetitions. We recorded RMSEs, biases and computation times for all simulations. As a benchmark, we used \texttt{fitMvdc} function in \texttt{copula} package. By default it maximizes the continuous-margin density likelihood. Hence, in our scenario, the likelihood will be misspecified. Generally, this is not a good option for a copula with discrete margins and should be avoided. For dimension $d=4$, we did not perform a simulation starting from empirical correlation and with numerical gradient, since the simulation starting from true tau and with numerical gradient took a long time to complete and we believe it would take even longer if one started from empirical correlation. 

\subsection{Implementation}

For the computation of the numerical gradient, we use built-in finite-difference approximations routine in $\Rstudio$ \texttt{optim} package. For all scenarios, to maximize the log-likelihood we use \texttt{optim} in $\Rstudio$ with \texttt{BFGS} optimization routine, which is a Newton-based routine. To compute the MVN cdf, by default we use the Miwa algorithm \citep{miwaEvaluationGeneralNonCentred2003} and resort to \citet{GenzBretzComputationMultivariateNormal2009} in case the covariance matrices are singular.

We are going to compare the following scenarios for the starting points. For the Poisson parameters, we are using the method-of-moments estimator for all the cases. Consider

\begin{align*}
  \bar{y}_j & = \frac{1}{n} \sum_{r=1}^{n} y_j^{(r)}, \quad j=1,\ldots,d.
\end{align*}

Then, by using the first moment,

\begin{equation}
  \label{eq:start-lam}
  \hat{\lambda}_{j}^{} = \bar{y}_j \Rightarrow \est{j}{s} = \log \left( \hat{\lambda}_{j}^{}\right)
\end{equation}

by using \autoref{eq:lam-reparam-tau}. For the dependence parameters, we consider two different scenarios.

As a first option we consider $\bp = \text{SCP} \left(\bz\right)$, the spherical reparametrization defined in Subsection~\ref{sec:reparam}. As the map is smooth and bijective, $\bz = \text{SCP}^{-1} \left(\bp\right)$. We then compute the empirical correlation

\begin{equation*}
  \label{eq:start-cor}
  \rst{ij}{c} = \widehat{\cor} \left(Y_i, Y_j\right) \text{ and } \zst{ij}{c} = \text{SCP}^{-1} \left(\rst{ij}{c}\right),
\end{equation*}

where the superscript $^{(c)}$ stands for empirical correlation. 

As a second option we use the method outlined in Option $1$ in Subsection~\ref{sec:option1}. Given $\hat{\lambda}_{j}^{}$ from \autoref{eq:lam-reparam-tau}, we use \autoref{eq:tau-eq-alld} to compute the starting values for spherically reparametrized correlation parameters. The general procedure is outlined in Algorithm~\ref{alg:sim}.

\begin{algorithm}
    \caption{Simulation algorithm}\label{alg:sim}
    \begin{algorithmic}
        \Require $n \geq 0, \bg = (\bl, \br), \mid \bg \mid := q = (d + d(d-1)/2)$
        \Ensure $\be = \log \left(\bl \right), \bz = \text{SCP}^{-1} \left(\br\right)$
        % Start timer
        \State  $t_0 \gets \text{\texttt{Sys.time()}}$ \Comment Start timer
        \For{$k=1,\ldots,N$}
          \State  $\byy^k \gets (\byy_1^{k}, \ldots, \byy_d^{k})$  \Comment Simulate from a Gaussian copula with Poisson margins
          \State $\hat{\bps}^{k,s} \gets (\boldsymbol{\hat{\eta}}^{(k,s)}, \hat{\boldsymbol{\zeta}}^{(k,s)})$ \Comment Compute the starting values
          \State $\hat{\bg}^k \gets \left( \exp \left( \hat{\boldsymbol{\eta}}^k \right), \text{SCP} \left( \hat{\boldsymbol{\zeta}}^k \right) \right)$ \Comment Store MLE parameters
        \EndFor
        \For{$l=1,\ldots, q$}
          \State  $\text{RMSE}_l \gets \sqrt{1/N \sum_{k=1}^{N} \left( \hat{\bg}_{k}^{(l)} - \bg^{(l)} \right)^2}$  \Comment Compute RMSE
          \State  $\text{Bias}_l \gets 1/N \sum_{k=1}^{N} \left( \hat{\bg}_{k}^{(l)} - \bg^{(l)} \right)$  \Comment Compute bias
        \EndFor
        \State  $t_1 \gets \text{\texttt{Sys.time()}}$ \Comment End timer
        \State  $\text{\texttt{Time}} \gets t_1 - t_0$ \Comment Record total time
        \State  $\text{\texttt{meanRMSE}} \gets 1/q \sum_{l=1}^{q} \text{RMSE}_l$ \Comment Record mean RMSE
        \State  $\text{\texttt{meanBias}} \gets 1/q \sum_{l=1}^{q} \text{Bias}_l$ \Comment Record mean bias
    \end{algorithmic}
\end{algorithm}

\subsection{Results}

The simulations were run according to the scenarios outlined in \autoref{tab:sim-res}, which also reports average RMSE, average bias multiplied by $10^3$ for all the parameters combined, in addition to total simulation running time. The simulations were performed on a personal Macbook mini with M2 processor and 16GB memory, using parallelized computation with $4$ cores. Parallelization here means shared-memory task parallelism on a single machine: we create a cluster of worker processes and distribute independent iterations to them. Specifically, we prepared a shared-memory worker pool using $\Rstudio$'s \texttt{parallel}/\texttt{doParallel} interface, capped at four logical cores for portability, and then executed the Monte Carlo loop sequentially. Each replicate $k = 1, \ldots, N$ is initialized with a deterministic seed \texttt{set.seed(k)} to ensure per-replicate reproducibility. 

\begin{table}[h]
    \centering
    \begin{tabular}{|c|c|c|c|c|c|c|r|r|}
        \hline 
        Setting & $d$ & Lambdas & Gradient & Start & Method & $\overline{\mathrm{RMSE}}$ & $\overline{\mathrm{Bias}} (\times 10^{3})$ & Time \\
        \hline 
$2.1$ & $2$ & $2, 3$ & Analytic & Corr & Reparam & $0.057$ & $0.002$ & $5.76S$ \\
        \hline
$2.2$ & $2$ & $2, 3$ & Analytic & Option 1 & Reparam & $0.055$ & $-0.163$ & $43.13S$ \\
        \hline
$2.3$ & $2$ & $2, 3$ & Numeric & Corr & Reparam & $0.054$ & $-0.138$ & $4.70S$ \\
        \hline
$2.4$ & $2$ & $2, 3$ & Numeric & Corr & fitMvdc & $0.197$ & $181.688$ & $11.98S$ \\
        \hline
$2.5$ & $2$ & $2, 3$ & Numeric & Option 1 & Reparam & $0.054$ & $-0.138$ & $53.19S$ \\
        \hline
$2.6$ & $2$ & $2, 3$ & Numeric & Option 1 & fitMvdc & $0.197$ & $181.720$ & $13.51S$ \\
        \hline
$2.7$ & $2$ & $0.5, 1$ & Analytic & Corr & Reparam & $0.048$ & $-3.921$ & $6.23S$ \\
        \hline
$2.8$ & $2$ & $0.5, 1$ & Analytic & Option 1 & Reparam & $0.041$ & $-0.146$ & $53.30S$ \\
        \hline
$2.9$ & $2$ & $0.5, 1$ & Numeric & Corr & Reparam & $0.041$ & $-0.033$ & $4.85S$ \\
        \hline
$2.10$ & $2$ & $0.5, 1$ & Numeric & Corr & fitMvdc & $0.425$ & $145.554$ & $16.10S$ \\
        \hline
$2.11$ & $2$ & $0.5, 1$ & Numeric & Option 1 & Reparam & $0.041$ & $-0.035$ & $53.11S$ \\
        \hline
$2.12$ & $2$ & $0.5, 1$ & Numeric & Option 1 & fitMvdc & $0.425$ & $145.555$ & $16.66S$ \\
        \hline
$4.1$ & $4$ & $0.6, 2, 4, 0.8$ & Analytic & Corr & Reparam & $0.054$ & $-2.233$ & $19H 43M 44.03S$ \\
        \hline
$4.2$ & $4$ & $0.6, 2, 4, 0.8$ & Analytic & Option 1 & Reparam & $0.050$ & $-0.821$ & $17H 8M 15.31S$ \\
        \hline
$4.3$ & $4$ & $0.6, 2, 4, 0.8$ & Numeric & Option 1 & Reparam & $0.065$ & $-2.139$ & $3d 20H 9M 21.88S$ \\
        \hline
    \end{tabular}
    \caption{Simulation results. The results show the usage of Analytic vs Numeric gradient, starting from Option $1$ vs empirical correlation ("Corr"), with unconstrained reparameterization method ("Reparam") vs fitMvdc.}
    \label{tab:sim-res}
\end{table}

% Comments on the tables
From \autoref{tab:sim-res}, we see the \texttt{fitMvdc} results are bad across all scenarios since RMSEs are the worst. This is expected because the likelihood is misspecified in that scenario. For our proposed method with the unconstrained parameters, results in dimension $2$ are very close across all scenarios, still starting from Option 1 and using analytical gradient generally yields the smallest RMSE. In dimension $4$, it becomes clear that the analytical gradient drastically improves the computation speed. Furthermore, starting from Option 1 reduces the bias as well as the computation time across all simulations. 

More detailed tables with RMSEs for each parameter individually are in Appendix~\ref{app:sim-tab}. The plots for all scenarios and parameters are available in Appendix~\ref{app:sim-res-plots}. We comment on them in the following paragraph. \autoref{tab:sim-res-d2} shows the more detailed per parameter RMSEs for $d=2$. Except for \texttt{fitMvdc}, the results are very good and generally close. The RMSEs for $\hat{\rho}$ are the best when starting from Option $1$ and with the analytical gradient for both moderate and small lambda settings. For $d=2$, full maximum likelihood with our unconstrained correlation parameterization ('Reparam') attains low RMSE across all parameters, both at moderate counts $(2,3)$ and in the low-count regime $(0.5,1)$. In contrast, the MVDC routine (\texttt{fitMvdc}) exhibits large errors for discrete data, with $\hat{\rho}$ particularly unstable at low counts. RMSEs under analytic and numeric gradients are virtually identical, confirming the correctness of our analytical score. We nonetheless use the analytical gradient for improved numerical stability. Starting from Option $1$ yields slightly better results than from the empirical correlation. The convergence rate, computed as the ratio of all converged optimizations to the total number of Monte Carlo runs $N$, was $1$.

\autoref{tab:sim-res-d4-lam} shows the more detailed per parameter RMSEs for $\lambda$ in $d=4$. We see analytic scores materially improve finite-sample efficiency relative to finite-difference (numeric) scores, with gains most pronounced on the parameters that are hardest to identify in the given regime ($\lambda_3$ here). Starting values matter, but once analytic gradients are in place, the gap between reasonable starts is second-order in this design. However, better starting values (Option $1$) does speed up the computation time significantly. Numerical gradients inject stochastic truncation error plus floating-point amplification. In non-smooth regions induced by piecewise rectangle sums, step-to-step score variability increases, so quasi-Newton updates inherit extra curvature noise, inflating estimator variance (\citet{NumericalOptimization2006}).

For $d=4$, RMSEs for $\rho_{ij}$ as shown in \autoref{tab:sim-res-d4-rho} are uniformly small with analytic gradients. Numeric gradients inflate errors, with the largest degradations on pairs involving index $3$ or $4$ (e.g., $\rho_{31}: 0.046$ for analytic gradient and $\rho_{31}: 0.073$ for numeric gradient approximation; $\rho_{43}: 0.054$ for analytic gradient and $\rho_{43}: 0.075$ for numeric gradient approximation). This is consistent with information loss from finite-difference noise carried through the inclusion-exclusion likelihood. Overall, the choice of gradient dominates the choice of start for dependence recovery. Analytic scores deliver stable, near-homogeneous performance across pairs, whereas numeric scores introduce heterogeneity and higher RMSE. However, the use of Option $1$ speeds up the convergence significantly and is therefore recommended.

% Comments on the plots
Violin plot for $\lambda, d=2$ and small intensities setting (\autoref{fig:violin-lam-d2-small}) shows that estimates are tightly centered on the truth (dashed line). Dispersion is larger for $\lo = 0.5$ than for $\lt = 1$, reflecting the smaller Fisher information at very low counts. Across different scenarios on the $x$-axis, centers remain stable, except for scenarios $2.10$ and $2.12$ (\texttt{fitMvdc}).

Violin plot for $\lambda, d=2$ and moderate intensities setting (\autoref{fig:violin-lam-d2-moderate}) shows that both parameters are essentially unbiased with visibly narrower violins than in the low-count case, consistent with increased information. Similarly to the low-intensity setting, scenarios $2.4$ and $2.6$ (\texttt{fitMvdc}) produce biased estimates.

Violin plot for $\lambda, d=4$ across three settings (\autoref{fig:violin-lam-d4}) shows that with analytic gradients (scenarios $4.1$ and $4.2$), all $\hat{\lambda}_j$ are well centered and similarly dispersed. Under numeric gradients (scenario $4.3$), violins show a very slight downward shifts for $\lt$ and $\lambda_4$, indicating variance inflation plus small bias. Furthermore, scenario $4.3$ (numeric gradient) presents with more outliers.

Violin plot for $\rho, d=2$ small intensities setting (\autoref{fig:violin-rho-d2-small}) shows that all scenarios except $2.10$ and $2.12$ (\texttt{fitMvdc}) centers track the truth closely. Scenarios $2.10$ and $2.12$ produce completely biased estimates.

Violin plot for $\rho, d=2$ moderate intensities setting (\autoref{fig:violin-rho-d2-moderate}) shows distributions contract relative to the low-count case, with means and medians aligned to the dashed line across dependence levels. Scenarios $2.4$ and $2.6$ (\texttt{fitMvdc}) produce biased estimates, but the situation looks better compared to the low-intensity case.

Violin plot for $\rho, d=4$ across three settings (\autoref{fig:violin-rho-d4}) shows analytic-gradient settings produce compact, nearly unbiased distributions across all pairs. The numeric-gradient setting shows clear variance inflation and pair-specific bias (most visible for $\rho_{31}$ and $\rho_{42}$), confirming that finite-difference noise degrades dependence recovery in this likelihood. \autoref{fig:violin-rho-d4} shows some bimodality for the estimates of correlation parameters when starting from empirical correlation. This is probably due to the fact that the algorithm gets sometimes "stuck" around the starting values. This observation is of importance and shows yet again the importance of carefully chosing the starting values. 

% Conclusion
To conclude, across dimensions and regimes, analytic gradients deliver near-unbiased, low-variance estimates for both margins and dependence parameters, while numeric gradients systematically inflate dispersion and induce small pair-specific biases, with effects most visible in $d=4$ and at lower count intensities. The above results motivate our recommendation to use analytic gradients and an unconstrained reparameterization on the real line for both margins and dependence as well as Kendall's tau-based start (Option $1$).

\section{Discussion}

We introduced three rank-based strategies to initialize Gaussian-copula dependence parameters from Kendall's tau in the discrete, low-count regime, and embedded them within an exact maximum-likelihood (ML) pipeline. Although our simulations emphasize the exact tau construction (Option~1), the approximation in Option~3a emerges as an appealing alternative when count intensities are moderate to high, where the simulations have shown a smallest bias, hence providing negligible loss in accuracy and starting values as close as possible to the ML estimators. Option~2 is not directly usable for computational gains in our experiments, yet it sheds some light on discrete copula constructions and, in our view, warrants further theoretical study.

A principal empirical finding is that analytic score functions improve both stability and speed of the Newton-type optimizer. This is especially pronounced as the dimension grows, where finite differences incur a prohibitive cost proportional to the number of free correlation parameters $p=d(d-1)/2$ and amplify Monte-Carlo noise in rectangle probabilities. We therefore recommend analytic gradients as the default from $d \geq 4$ onward.

Although likelihood evaluation requires $d$-variate multivariate normal (MVN) rectangle probabilities, the inclusion-exclusion expansion is sparse for low counts. For an observation $\by^r = (\byrs)$, let $s(y)=\sum_{j=1}^d \mathbf{1}\{y_j^{(r)}>0\}$ denote the number of strictly positive entries. Then only $2^{s(y)}$ of the $2^d$ terms are nonzero, because any term involving $F_{Y_j}(y_j-1)=0$ vanishes when $y_j=0$. In low-intensity regimes $s(y)$ is typically much smaller than $d$, so the effective number of MVN calls per observation is drastically reduced. Coupled with an unconstrained angle-based parameterization of the correlation factor and analytic gradients, this yields practical scalability beyond the bivariate case. 

Our implementation supports arbitrary $d$. We documented detailed results for $d=2$ and $d=4$ and provided a run-time study illustrating how zero inflation stabilizes computational cost as $d$ increases.

\textbf{Statistical implications.}
Under correct specification, full ML preserves its classical guarantees (consistency, asymptotic normality, and efficiency), while continuous extension and pseudo/composite surrogates generally do not. The proposed tau-based initializers are not substitutes for ML but rather devices to reduce iterations and improve conditioning. In the regimes we considered, the empirical mapping from the copula parameter to our discrete Kendall's tau remains monotone and well behaved, enabling informative parameter starts. A formal characterization of monotonicity and local curvature for small means is an important direction for theory. Our simulations also highlight that methods developed for copulas with continuous margins can be substantially biased in the low-count setting, and should not be used.

\textbf{Computational complexity and scaling.}
Per iteration, the dominant cost is the evaluation of MVN rectangles and their derivatives. The linear algebra for the spherical-Cholesky map ($\lrep \mapsto \bp= \lrep \lrep^{\top}$) is $\mathcal{O}(d^3)$ but negligible relative to the cdf workload, and the reparameterization enforces positive-definiteness without constraints, improving line-search robustness. In $d=2$, numeric gradients can be competitive. Our simulations show analytical gradients take longer to compute. However, by $d=4$, analytic scores are clearly preferable in both accuracy and computational time.

\textbf{Guidance on initializer choice.}
The empirical correlation-based start provides clearly biased estimates in all dimensions. From $d \geq 4$, the tau-based Option~1 reduces time-to-convergence and slightly lowers RMSE, reflecting a better match to the likelihood's local curvature. Option~3a becomes particularly attractive as means rise into the moderate-count regime, where its approximation error fastly converges towards zero while its cost remains low.

\textbf{Limitations and diagnostics.}
Two limitations are worth stating explicitly. First, at extremely low intensities the information about dependence can be weak (many all-zero vectors), which inflates finite-sample variance. Hence, profiling the log-likelihood along selected dependence directions and reporting curvature diagnostics is advisable. Second, when dependence is very strong, the MVN probabilities become numerically delicate. In our scenarios, we simulated a mixture of moderate and strong correlations which yielded a convergence rate of $1$ across all simulations.

\textbf{Extensions.}
Several extensions follow naturally. First, if dependence is sparse or block-structured, one can impose a sparsity pattern in a modified Cholesky factor or adopt block-diagonal $\lrep$, thereby reducing both parameter count and the effective cdf dimension per block. Next, the same pipeline applies to zero-inflated Poisson (ZIP), hurdle, and negative binomial margins, with tau-based starts defined by the corresponding discrete tau. Analytical gradients can be adapted to accomodate different margins. Other techniques such as Automatic Differentiation (AD) could also be explored \citep{kristensenTMBAutomaticDifferentiation2016}. Observed-information (outer-product and Hessian) estimators are straightforward with analytic scores; sandwich corrections enable robustness checks against mild marginal misspecification. Finally, the initialization logic extends to other elliptical copulas and to vines, where pairwise tau-based starts are standard; our discrete-tau constructions can be applied directly to these settings.

To summarize, the combination of discrete-aware tau-based initialization, exact rectangle probabilities, and analytic scores provides a pragmatic route to retain the statistical benefits of ML while remaining computationally tractable for moderate and, in many applications, higher dimensions. Our results at $d=2$ and $d=4$ demonstrate near-unbiased estimation with low RMSE and materially improved convergence when tau-based starts are used in higher dimension; the same design principles scale as data sparsity increases and as structure in the correlation is exploited.

% \section{Appendix}

\begin{appendices}
    % Appendix A
\section{Proofs}
\label{app:proofs}

\subsection{Skellam distribution}
\label{app:skellam}
\begin{definition}[Skellam distribution, \citep{skellamFrequencyDistributionDifference1946}]
    \label{def:skellam}

Let $X, Y$ follow independent Poisson distributions with parameters $\lx, \ly$. Then, $D:= X-Y$ follows a Skellam distribution with parameters $\lx, \ly$, with the pmf given by

\begin{equation*}
    \label{eq:skellam-pmf}
    p\left(d ; \lx, \ly\right)=\pr{D=d}=e^{-\left(\lx+\ly\right)}\left(\frac{\lx}{\ly}\right)^{d / 2} \mathcal{I}_d\left(2 \sqrt{\lx \ly}\right),
\end{equation*}

where $\mathcal{I}_d$ is the modified Bessel function of the first kind of order $d$. In the special case when $X,Y$ are i.i.d., i.e. Poisson disributed with parameter $\lambda$, the Skellam distribution simplifies to

\begin{equation}
    \label{eq:skellam-pmf-i.i.d.}
    p\left(d ; \lambda\right)=\pr{D=d}=e^{-\left(2 \lambda\right)} \mathcal{I}_d\left(2 \lambda\right).
\end{equation}

For any integer $k$, its cdf is obtained by summing the pmf:

\begin{equation*}
    \label{eq:skellam-cdf}
    F(k;\lx,\ly) = \pr{D\le k} = \sum_{j=-\infty}^{k} e^{-(\lx+\ly)}\left(\frac{\lx}{\ly}\right)^{j/2} \mathcal{I}_{|j|}\Bigl(2\sqrt{\lx\ly}\Bigr).
\end{equation*}
\end{definition}

\subsection{Probability of univariate Poisson ties}
\label{app:pois-ties}
\begin{result}[Corrected Kendall's tau for a copula with Poisson marginals, $d=2$.]
    \label{app:pois-univ-tie}
    Let $X \sim \mathrm{Po} (\lambda)$ and $X_1, X_2$ be two independent copies from $\mathrm{Po}(\lambda)$. Then the probability of ties for $X$ is given by

    \begin{equation}
        \label{eq:pois-univ-tie}
        \begin{aligned}
            \pr{X_1=X_2} &= \sum_{x=0}^{\infty} \pr{X_1 = X_2 | X_2 =x} \pr{X_2=x} = \sum_{x=0}^{\infty} \pr{X_1=x} \pr{X_2=x} = \\
            &= \sum_{x=0}^{\infty} \left( \pr{X=x}\right)^2 = e^{-2 \lambda} \sum_{x=0}^{\infty} \left(\frac{\lambda^{2x}}{(x!)^2}\right) = e^{-2 \lambda} \mathcal{I}_0(2 \lambda),
        \end{aligned}
    \end{equation}

    where $\mathcal{I}_0(z)$ is the modified Bessel function of the first kind of order $0$. The result holds due to Equation $(9.1.10)$ in \citet{AbramowitzHandbookMathematical}. Indeed, by setting $\nu = 0$, we have that $\mathcal{I}_0 (z) = \sum_{k=0}^{\infty } \frac{ \left( - 1/4 z^2\right)^k}{k! \Gamma \left(k+1 \right)} = \sum_{k=0}^{\infty } \frac{ \left( - 1/4 z^2\right)^k}{ \left(k!\right)^2}$.
\end{result}

\subsection{Proof of \autoref{prop:tau-skellam}}
\label{app:tau-skellam}

To prove \autoref{prop:tau-skellam}, the following Lemmas will be useful and must be introduced first. 

\begin{lemma}
    \label{lem:conv-copula-discrete}
    Let $(X, Y)$ have Poisson margins $\mathrm{Po}\left(\lx\right), \mathrm{Po}\left(\ly\right)$ and consider a copula $C_\rho$ (not unique on $\ran\left(F_X\right) \times \ran\left(F_Y\right)$), where $\ran$ is the range. For two i.i.d. copies $\left(X_1, Y_1\right)$, $\left(X_2, Y_2\right)$ of $(X,Y)$, define the differences $D_X=X_2-X_1$ and $D_Y=Y_2-Y_1$. We have that

    $$
    D_X \sim \operatorname{Skellam}\left(\lx, \lx\right), \quad D_Y \sim \operatorname{Skellam}\left(\ly, \ly\right),
    $$

    and there exists a (unique on the grid) subcopula $\cs$ on $\ran\left(F_{D_X}\right) \times \ran\left(F_{D_Y}\right) \subset[0,1]^2$ such that

    $$
    H\left(d_X, d_Y; \rs, \lx, \ly\right)=\pr{D_X \leq d_X, D_Y \leq d_Y} =\cs\left(F_{D_X}\left(d_X\right), F_{D_Y}\left(d_Y\right)\right)
    $$

    for any $d_X$ and $d_Y$ in $\Z$.
\end{lemma}
 
To prove Lemma~\ref{lem:conv-copula-discrete}, we will use the following definition from \citet{nelsenIntroductionCopulas2006}.

\begin{definition}[Definition $2.2.1$ in \citet{nelsenIntroductionCopulas2006}]
    \label{def:nelsen}
    A two-dimensional subcopula (or 2-subcopula, or briefly, a subcopula) is a function $\cs$ with the following properties:
    \begin{enumerate}
        \item Dom $\cs=S_1 \times S_2$, where $S_1$ and $S_2$ are subsets of I containing 0 and 1 ;
        \item $\cs$ is grounded and 2-increasing;
        \item For every $u$ in $S_1$ and every $v$ in $S_2$,
    \end{enumerate}

    $$
    \cs(u, 1)=u \text { and } \cs(1, v)=v .
    $$
\end{definition}

\begin{proof}[Proof of Lemma \ref{lem:conv-copula-discrete}]
    For the marginal laws of the differences $D_X$ and $D_Y$, note that $X_1$ and $X_2$ are i.i.d. $\mathrm{Po} (\lx)$, therefore $D_X = X_2 - X_1$ is $\text{Skellam}(\lx,\lx)$ with an analogous statement for $D_Y$. 

    To prove that there exist a subcopula, we will show that \autoref{def:nelsen} is satisfied. Namely, that the domain of a subcopula is a subset of $[0,1]^2$, that a subcopula is grounded and $2$-increasing, and that for every $u$ in the range of $D_X$ and $v$ in the range of $D_Y$, $C(u,1) = u$, $C(1,v) = v$.
    
    Since $D_X$ and $D_Y$ are obtained from two independent copies of a bivariate random vector $(X, Y)$ with discrete count-marginals and joint dependence given by a copula $C_\rho$, the convolution of the random variables can itself be decomposed via a copula. Let

$$
H\left(d_X, d_Y; \rs, \lx, \ly\right)=\pr{D_X \leq d_X, D_Y \leq d_Y},
$$

and denote its marginal cdf's by

$$
F_{D_X}\left(d_X\right)= \pr{D_X \leq d_X} \quad \text { and } \quad F_{D_Y}\left(d_Y\right)= \pr{D_Y \leq d_Y},
$$

then by Sklar's theorem there exists a function

$$
\cs:[0,1]^2 \rightarrow[0,1],
$$

such that

$$
H_{ }\left(d_X, d_Y; \rs, \lx, \ly\right)=\cs\left(F_{D_X}\left(d_X\right), F_{D_Y}\left(d_Y\right)\right).
$$

Define $\cs$ as

$$
\cs(u, v):=H_{ }\left(F_{D_X}^{-1}(u), F_{D_Y}^{-1}(v); \rs, \lx, \ly\right).
$$

Note that for every $(u, v)$ in $\dom (\cs), 0 \leq \cs(u, v) \leq 1$, so that $\ran (\cs)$ is also a subset of $[0,1]$, with $\dom$ being the domain.

Now we will show \autoref{def:nelsen} holds.

\begin{enumerate}
    \item Let

    $$
    S_1=\left\{F_{D_X}\left(d_X\right): d_X \in \mathbb{Z}\right\} \quad \text { and } \quad S_2=\left\{F_{D_Y}\left(d_Y\right): d_Y \in \mathbb{Z}\right\}
    $$
    
    so that $S_1$ and $S_2$ are the ranges of the marginal cdf's of $D_X$ and $D_Y$, respectively. Because every cdf $F$ satisfies
    
    $$
    \lim _{x \rightarrow-\infty} F(x)=0 \quad \text { and } \quad \lim _{x \rightarrow+\infty} F(x)=1,
    $$
    
    we have that $0,1 \in S_1$ and $0,1 \in S_2$.

    \item \begin{enumerate}[(a)]
        \item Groundedness. Since for any $d_X$ we have

        $$
        F_{D_Y}(-\infty)=0 \quad \Longrightarrow \quad H_{ }\left(d_X,-\infty; \rs, \lx, \ly \right)=0
        $$

        it follows that for any $u \in[0,1]$, with $d_X=F_{D_X}^{-1}(u)$,

        $$
        \cs(u, 0)=H_{ }\left(F_{D_X}^{-1}(u), F_{D_Y}^{-1}(0); \rs, \lx, \ly\right)=H_{ }\left(F_{D_X}^{-1}(u),-\infty; \rs, \lx, \ly\right)=0.
        $$

        A similar argument shows that $\cs(0, v)=0$ for all $v$.

        \item $2$-increasingness. A function $C$ is said to be 2-increasing if, for all $u_1 \leq u_2$ and $v_1 \leq v_2$ in $[0,1]$, the "volume"

        $$
        \Delta=C\left(u_2, v_2\right)-C\left(u_2, v_1\right)-C\left(u_1, v_2\right)+C\left(u_1, v_1\right)
        $$

        satisfies $\Delta \geq 0$ \citep[Definitions $2.1.1$ and $2.1.2$]{nelsenIntroductionCopulas2006}.
        Because $H_{ }$ is a joint cdf of $\left(D_X, D_Y\right)$, it is 2-increasing by definition. More precisely, if we set

        $$
        d_{X_1}=F_{D_X}^{-1}\left(u_1\right), \quad d_{X_2}=F_{D_X}^{-1}\left(u_2\right), \quad d_{Y_1}=F_{D_Y}^{-1}\left(v_1\right), \quad d_{Y_2}=F_{D_Y}^{-1}\left(v_2\right)
        $$

        then

        $$
        \Delta=H_{ }\left(d_{X_2}, d_{Y_2}; \rs, \lx, \ly\right)-H_{ }\left(d_{X_2}, d_{Y_1}; \rs, \lx, \ly\right)-H_{ }\left(d_{X_1}, d_{Y_2}; \rs, \lx, \ly\right)+H_{ }\left(d_{X_1}, d_{Y_1}; \rs, \lx, \ly\right) \geq 0.
        $$

        Thus, the function $\cs$ inherits the 2-increasing property from $H_{ }$.
    \end{enumerate}
    \item Note that
    
$$
F_{D_Y}(+\infty)=1 \quad \text { so that } \quad H_{ }\left(d_X,+\infty; \rs, \lx, \ly\right)=F_{D_X}\left(d_X\right).
$$

Thus, for any $u \in[0,1]$ with $d_X=F_{D_X}^{-1}(u)$ we have

$$
\cs(u, 1)=H_{ }\left(F_{D_X}^{-1}(u), F_{D_Y}^{-1}(1); \rs, \lx, \ly\right)=H_{ }\left(F_{D_X}^{-1}(u),+\infty; \rs, \lx, \ly\right)=F_{D_X}\left(F_{D_X}^{-1}(u)\right)=u.
$$

Similarly, $\cs(1, v)=v$ for all $v$.

\end{enumerate}

Using Definition $2.2.1$ from \citet{nelsenIntroductionCopulas2006}, we have therefore that $\cs$ is a subcopula on $\ran\left(F_{D_X}\right) \times \ran\left(F_{D_Y}\right) \subset[0,1]^2$  as it satisfies all the required properties.

\end{proof}

\begin{lemma}
    \label{lem:tie-terms}
    The tie-corrected Kendall's $\tau$ (as in \autoref{eq:nik-tau}) for a copula with discrete margins can be rewritten using the distributions of the differences of the margins. Specifically, for Poisson margins with $B_j(\lambda_j)= \pr{D_j=0} = e^{-2 \lambda_j} I_0\left(2 \lambda_j\right)$ for $j= \left\{X,Y \right\}$, the tie-corrected Kendall's $\tau_S$ with Skellam margins, built upon the Kendall's tau for a copula with Poisson margins, satisfies

    $$
    \tau_S(X, Y)=4 \pr{D_X \leq-1, D_Y \leq-1} - 1 + B_X(\lx) + B_Y(\ly) - \pr{D_X=0, D_Y=0}.
    $$
\end{lemma}
\begin{proof}[Proof of Lemma \ref{lem:tie-terms}]
    First, by applying the definition of Skellam pmf as in \autoref{eq:skellam-pmf-i.i.d.} and the definition of $B_j$ as in \autoref{eq:pois-ties},

    \begin{align*}
        B_j(\lambda_j) = \pr{D_j = 0} = \posq{j}.
    \end{align*}

    By using the identity as in \citet[Equation (4)]{denuitConstraintsConcordanceMeasures2005}, we have

    \begin{align*}
        &\tau_{S} (X,Y) = 4 \pr{X_2<X_1, Y_2<Y_1} - 1 + \pr{X_2=X_1 \quad \text { or } \quad Y_2=Y_1} = \\
        & = 4 \pr{X_2- X_1<0, Y_2-Y_1<0} - 1 + \pr{X_2-X_1 = 0} + \pr{Y_2-Y_1=0} - \pr{X_2-X_1=0, Y_2-Y_1=0} = \\
        & = 4 \pr{X_2- X_1 \leq -1, Y_2-Y_1 \leq -1} - 1 + \pr{X_2-X_1 = 0} + \pr{Y_2-Y_1=0} - \pr{X_2-X_1=0, Y_2-Y_1=0} = \\
        & = 4 \pr{D_X \leq -1, D_Y \leq -1} - 1 + \pr{D_X = 0} + \pr{D_Y=0} - \pr{D_X=0, D_Y=0} = \\
        & = 4 \pr{D_X \leq -1, D_Y \leq -1} - 1 + \posq{X} + \posq{Y} - \pr{D_X=0, D_Y=0}.
    \end{align*}
\end{proof}

\begin{lemma}
    \label{lem:gaus-cop-monot}
    Fix $a, b \in \mathbb{R}$. For the standard bivariate normal cdf $\Phi_2(a, b ; \rho)$, the map $\rho \mapsto \Phi_2(a, b ; \rho)$ is continuous and strictly increasing on $(-1,1)$.
\end{lemma}

\begin{proof}[Proof of Lemma \ref{lem:gaus-cop-monot}]
    This is a classical result that follows from Slepian's inequality, see Theorem $5.1.7$ in \citet{tongMultivariateNormalDistribution2012}.
\end{proof}

\begin{lemma}
    \label{lem:cor-discrete}
    Let $F_X, F_Y$ be marginal cdf's (continuous or discrete) with finite second moments, $\sigma_X = \sqrt{\var (X)}, \, \sigma_Y = \sqrt{\var (Y)}$ and define $g_X(z)=F_X^{-1}\{\Phi(z)\}, g_Y(z)=F_Y^{-1}\{\Phi(z)\}$, with $\Phi$ being the standard normal univariate cdf. For $\left(Z_1, Z_2\right) \sim N\left(0, \Sigma_{\rho}\right)$ with $\Sigma_\rho=\begin{psmallmatrix}1&\rho\\ \rho&1\end{psmallmatrix}$ and correlation $\rho \in (-1,1)$, set $X^{}=g_X\left(Z_1\right), Y^{}=g_Y\left(Z_2\right)$ and define

    $$
    \psi(\rho)=\cor\left(X^{}, Y^{}\right) .
    $$

    Then:

    \begin{enumerate}[(i)]
        \item $\psi$ admits the Hermite-Price power series $\psi(\rho)=\sum_{k \geq 1} b_k \rho^k$ with $b_k= \left\{a_k\left(F_X\right) a_k\left(F_Y\right)\right\} /\left(k ! \sigma_X \sigma_Y\right)$ and $a_k(F)=\int_{-\infty}^{\infty}F^{-1} (\Phi (t)) H_k(t) \varphi(t) d \,t$, with $H_k(t)$ the probabilists' Hermite polynomial of degree $k$ and $\varphi$ the standard normal pdf. 
        \item $\psi$ is real-analytic on $(-1,1)$, continuous on $[-1,1]$, strictly increasing, $\psi(0)=0$.
        \item $\psi(\rho) \equiv \rho$ iff both marginals are Gaussian up to affine transforms. Otherwise $\psi$ is not identity in general.
        \item $\lim _{\rho \rightarrow \pm 1} \psi(\rho)$ are the Fréchet bounds compatible with $F_X, F_Y$ and are attained by the Gaussian copula model.
    \end{enumerate}
\end{lemma}

\begin{proof}[Proof of Lemma \ref{lem:cor-discrete}]
    Below we prove the points (i)-(iv).
    \begin{enumerate}[(i)]
        \item By their definition, $X^{} \sim F_X, Y^{} \sim F_Y$, and their joint distribution has Gaussian copula parameter $\rho$. Define for any cdf $F$,

        \begin{equation}
            \label{eq:ak}
            a_k(F) = \ex{g(Z) H_k(Z)} = \int_{-\infty}^{\infty}F^{-1} (\Phi (t)) H_k(t) \varphi(t) d \,t,
        \end{equation}

        with $H_k(t)$ the probabilists' Hermite polynomial of degree $k$. Then by Lemma $2.1$ in \citet{hanCorrelationStructureGaussian2016}, the covariance has the absolutely convergent Hermite series

        $$
        \cov\left(X^{}, Y^{}\right)=\sum_{k=1}^{\infty} \frac{a_k\left(F_X\right) a_k\left(F_Y\right)}{k!} \rho^k, \quad|\rho|<1 .
        $$

        Dividing by $\sigma_X \sigma_Y$, we obtain the map

        $$
        \psi(\rho) =\cor\left(X^{}, Y^{}\right)=\sum_{k=1}^{\infty} b_k \rho^k, \quad b_k=\frac{a_k\left(F_X\right) a_k\left(F_Y\right)}{k!\sigma_X \sigma_Y} .
        $$

        \item This representation immediately yields continuity and real-analyticity on $(-1,1)$. It is also easy to verify that $\psi(0) = 0$. The covariance depends on $\rho$ through $\ex{XY} = \ex{g_X(Z_1) g_Y(Z_2)}$. Since $g_X, g_Y$ are nondecreasing as they are quantiles, they admit a Lebesgue-Stieltjes measure $\mu_X, \mu_Y$ with $g_X(z)=g_X(-\infty)+\int \mathbf{1}\{t \leq z\} d \,\mu_X(t) = \int \mathbf{1}\{t \leq z\} d \,\mu_X(t)$ and likewise for $g_Y$. Then

    \begin{align*}
        \ex{g_X(Z_1) g_Y(Z_2)} &= \iint g_X(z_1) g_Y(z_2) d \,\Phi_2(z_1, z_2; \rho) = \\
        & = \iint \left( \iint \mathbf{1}\{t \leq z_1\} \mathbf{1}\{s \leq z_2\} d \,\mu_X(t)  d \,\mu_Y(s)\right) d \,\Phi_2(t, s; \rho) = \\
        & = \iint \left( \iint \mathbf{1}\{t \leq z_1\} \mathbf{1}\{s \leq z_2\} d \,\Phi_2(t, s; \rho) \right) d \,\mu_X(t)  d \,\mu_Y(s) = \\
        & = \iint \ex{\mathbf{1}\{t \leq z_1\} \mathbf{1}\{s \leq z_2\}} d \,\mu_X(t)  d \,\mu_Y(s) = \iint \pr{Z_1 \geq t, Z_2 \geq s} d \,\mu_X(t)  d \,\mu_Y(s) = \\
        & = \iint \left(1-\Phi(t) - \Phi(s) + \Phi_2(t, s; \rho)\right) d \,\mu_X(t)  d \,\mu_Y(s),
    \end{align*}

    by using the Lebesgue-Stieltjes representation and Fubini's theorem. We also have that 

    \begin{align*}
        \frac{\partial}{ \partial \rho} \ex{g_X(Z_1) g_Y(Z_2)} & = \frac{\partial}{ \partial \rho} \iint \left(1-\Phi(t) - \Phi(s) + \Phi_2(t, s; \rho)\right) d \,\mu_X(t)  d \,\mu_Y(s) = \\
        & = \iint \frac{\partial}{\partial \rho} \Phi_2(t, s; \rho) d \,\mu_X(t)  d \,\mu_Y(s) = \iint \varphi_2(t, s; \rho) d \,\mu_X(t)  d \,\mu_Y(s) \geq 0,
    \end{align*}

    by Plackett's identity \citep{plackettReductionFormulaNormal1954} and positivity of the Stieltjes measures. Thus $\psi^{\prime}(\rho) \geq 0$ on $(-1,1)$, with strict increase unless a margin is degenerate. 

    \item When both margins are Gaussian (up to affine transforms), only the $k=1$ term of $a_k(\Phi)$ survives, so $\psi(\rho)=\rho$. This follows immediately by plugging the Gaussian quantile maps into \autoref{eq:ak}. Indeed, in case of Gaussian margins with means $\mu_X, \mu_Y$ and variances $\sigma_X^2, \sigma_Y^2$, we have for each $\mu \in \left\{ \mu_X, \mu_Y\right\}, \sigma \in \left\{ \sigma_X, \sigma_Y\right\}$
    
    \begin{align*}
        g(z) & = \mu + \sigma z, \\
        a_k (\Phi) & = \ex{g(Z)H_k(Z)} = \ex{(\mu + \sigma Z)H_k(Z)} = \mu \ex{H_k(Z)} + \sigma \ex{Z H_k(Z)} = \sigma \ex{Z H_k(Z)}.
    \end{align*}
    
    First, we used a result from \citet{plucinskaPolynomialNormalDensities1998} which states $\ex{H_k(Z)}=0, k \ge 1$. Next, by using Lemma $3.6.5$ in \citet{CasellaStatisticalInference} (Stein's Lemma), $\ex{Z H_k(Z)} = \ex{H_k' (Z)}$. 
    
    Hence, for $k=1$ and on, we obtain
    
    \begin{align*}
        a_1(\Phi) &= \sigma \ex{H_1' (Z)} = \sigma \ex{1} = \sigma , \\
        a_2(\Phi) &= \sigma \ex{H_2' (Z)} = 2 \sigma \ex{Z}= 0, \\
        a_3(\Phi) &= \sigma \ex{H_3' (Z)} = \sigma \left( 3 \ex{Z^2} - 3 \right)  = 0, \\
        \vdots \\
    \end{align*}

    Hence,

    \begin{align*}
        \psi(\rho) &= b_1 \rho = \frac{a_1 (\Phi_X) a_1 (\Phi_Y)}{\sigma_X \sigma_Y}  \rho = \rho.
    \end{align*}
    
    Conversely, if $\psi(\rho) \equiv \rho$ on $(-1,1)$, all $k \geq 2$ coefficients must vanish, forcing each quantile to be affine (hence the margins should be Gaussian). Indeed, one can decompose

    \begin{align*}
        \psi(\rho) &= b_1 \rho + \sum_{k=2}^{\infty} b_k \rho^k.
    \end{align*}
    
    By Cauchy-Schwarz, $ \left|a_1 (F) \right| = \left|\ex{g(Z) Z} \right| = \left| \cov (g(Z), Z) \right| \le \sqrt{ \var \left( g(Z) \right) \var \left( Z\right)} = \sigma$, with equality if and only if $g$ is affine in $Z$, i.e. iff the margin is Gaussian up to an affine transform. Thus $\psi (\rho) \equiv \rho$ holds if and only if both $g_X$ and $g_Y$ are affine, i.e., both margins are Gaussian up to affine transforms. Furthermore, for arbitrary continuous margins, higher-order coefficients are nonzero and $\psi$ is not identity in general (see Remark on p. 61 of \citet{hanCorrelationStructureGaussian2016}). 

    \item Finally, in the bivariate case the Gaussian copula tends to the comonotone/countermonotone copula as $\rho \rightarrow \pm 1$ respectively, so the limits are the Fréchet-Hoeffding bounds and are achieved by the Gaussian-copula family \citep{hanCorrelationStructureGaussian2016}.
    \end{enumerate}

\end{proof}

Now we are ready to prove \autoref{prop:tau-skellam}.
\begin{proof}[Proof of \autoref{prop:tau-skellam}]
    Let $\hs$ denote the joint cdf of $(D_X, D_Y)$ and, for a chosen parametric copula family (e.g., Gaussian), write

    $$
    \hs\left(d_X, d_Y ; \rs, \lx, \ly\right)=\cs\left(F_{D_X}\left(d_X\right), F_{D_Y}\left(d_Y\right) ; \rs\right).
    $$

    By invoking Lemma \ref{lem:conv-copula-discrete}, $\cs$ is a subcopula. Define

    \begin{align*}
            & \as := 4 \hs\left(-1,-1 ; \rs, \lx, \ly\right) \\
            & \quad - \left[\hs\left(0,0 ; \rs, \lx, \ly\right)-\hs\left(-1,0 ; \rs, \lx, \ly\right)-\hs\left(0,-1 ; \rs, \lx, \ly\right)+\hs\left(-1,-1 ; \rs, \lx, \ly\right)\right].
    \end{align*}

    Under a Gaussian copula for $\left(D_X, D_Y\right)$,

    $$
    \hs\left(d_X, d_Y ; \rs, \lx, \ly\right)=\Phi_2\left(\Phi^{-1}\left(F_{D_X}\left(d_X\right)\right), \Phi^{-1}\left(F_{D_Y}\left(d_Y\right)\right) ; \rs\right),
    $$

    so each term in $\as$ is continuous and strictly increasing in $\rs$ by Lemma \ref{lem:gaus-cop-monot}, hence so is $\as$.

    Let

    $$
    A\left(\rho , \lx, \ly\right):=4 \ex{H(X-1, Y-1; \rho, \lx, \ly)} - \ex{h(X, Y; \rho, \lx, \ly)},
    $$

    with $H, h$ computed under the original Poisson margins and copula parameter $\rho$. Then $\tau(X, Y)=-1+B_X(\lx)+ B_Y(\ly)+A\left(\rho , \lx, \ly\right)$ as in \autoref{eq:nik-tau-ab} and by using Lemma \ref{lem:tie-terms}. Because $\as$ is continuous and strictly increasing and its image spans an interval containing the target value $A\left(\rho , \lx, \ly\right)$, the intermediate value theorem (e.g., \citet[Section $4.5$]{abbottUnderstandingAnalysis2015}) yields existence and uniqueness of $\rs$ such that

    $$
    \as=A\left(\rho , \lx, \ly\right), \quad \text { i.e. } \quad \tau(X, Y)=-1+B_X(\lx)+B_Y(\ly)+\as .
    $$

    Finally, by using Lemma \ref{lem:cor-discrete}, we write the induced Pearson correlations maps as

    $$
    \psi_{\mathrm{Po}}(\rho):=\cor(X, Y), \quad \psi_{\mathrm{Sk}}\left(\rs\right):=\cor\left(D_X, D_Y\right).
    $$

    A direct calculation using independence of copies gives

    $$
    \psi_{\mathrm{Sk}}\left(\rs\right) = \cor\left(D_X, D_Y\right)=\frac{\cov\left(X_2-X_1, Y_2-Y_1\right)}{\sqrt{2 \lx} \sqrt{2 \ly}}=\frac{\cov(X, Y)}{\sqrt{\lx \ly}}=\psi_{\mathrm{Po}}(\rho),
    $$

    so $\rs= \psi_{\mathrm{Sk}}^{-1} \left( \cor\left(D_X, D_Y\right)\right)  = \psi_{\mathrm{Sk}}^{-1}\left(\psi_{\mathrm{Po}}(\rho)\right)$. Because the difference map on the latent Gaussian scale is linear (and would preserve $\rho$), whereas the observed transformation $z \mapsto F_j^{-1}(\Phi(z))$ is non-linear and non-monotone for discrete margins, in general $\rs \neq \rho$.

    By using Lemma \ref{lem:cor-discrete}, we write 

    \begin{align*}
        &a_k(F)=\int_{-\infty}^{\infty} F^{-1}(\Phi(t)) H_k(t) \varphi(t) d \,t. \\     
    \end{align*}

Hence

    \begin{align*}
        \psi_{\mathrm{Po}}(\rho)=\cor(X, Y) &= \frac{1}{\sqrt{\lx \ly}} \sum_{k=1}^{\infty} \frac{a_k\left(F_X\right) a_k\left(F_Y\right)}{k!} \rho^k, \\
        \psi_{\mathrm{Sk}}(\rs)=\cor(D_X, D_Y) &= \frac{1}{2 \sqrt{\lx \ly}} \sum_{k=1}^{\infty} \frac{a_k\left(F_{D_X}\right) a_k\left(F_{D_Y}\right)}{k!} (\rs)^k.
    \end{align*}

    There is no closed form, but the series converge rapidly \citep{hanCorrelationStructureGaussian2016}). 

    Finally, $\rs$ is the solution to 

    \begin{equation}
        \sum_{k=1}^{\infty} \frac{a_k\left(F_X\right) a_k\left(F_Y\right)}{k!} \rho^k = \frac{1}{2 } \sum_{k=1}^{\infty} \frac{a_k\left(F_{D_X}\right) a_k\left(F_{D_Y}\right)}{k!} (\rs)^k,
    \end{equation}

    therefore, $\rs = \psi_{\mathrm{Sk}}^{-1}\left(\psi_{\mathrm{Po}}\left(\rho\right)\right)$. As $\psi_{Sk} \neq \psi_{\mathrm{Po}}$, we will generally have that $\rho \neq \rs $. However, since both maps converge to identity as $\lx, \ly \to \infty$, we have that $\rho - \rs \to 0$ in this limit.
\end{proof}

\subsection{Additional details on Option 2}
\label{app:option2}

To illustrate \autoref{prop:tau-skellam}, consider a grid of lambdas with $\boldsymbol{\lambda} = \{0.05, 0.1, 0.5, 1, 2, 3, 4, 5 \}$ and $\boldsymbol{\rho} = \{-0.8, -0.5, -0.2, 0.2, 0.5, 0.8 \}$ and define the loss function

$$
\mathcal{L} (\rs, \rho, \lx, \ly) := \left(\as- A(\rho, \lx, \ly) \right)^2.
$$

Then, consider

$$
\rs = \underset{\theta}{\argmin} \, \mathcal{L} (\theta, \rho, \lx, \ly).
$$

\begin{figure}[t!]
    \centering
    \includegraphics[width=\textwidth]{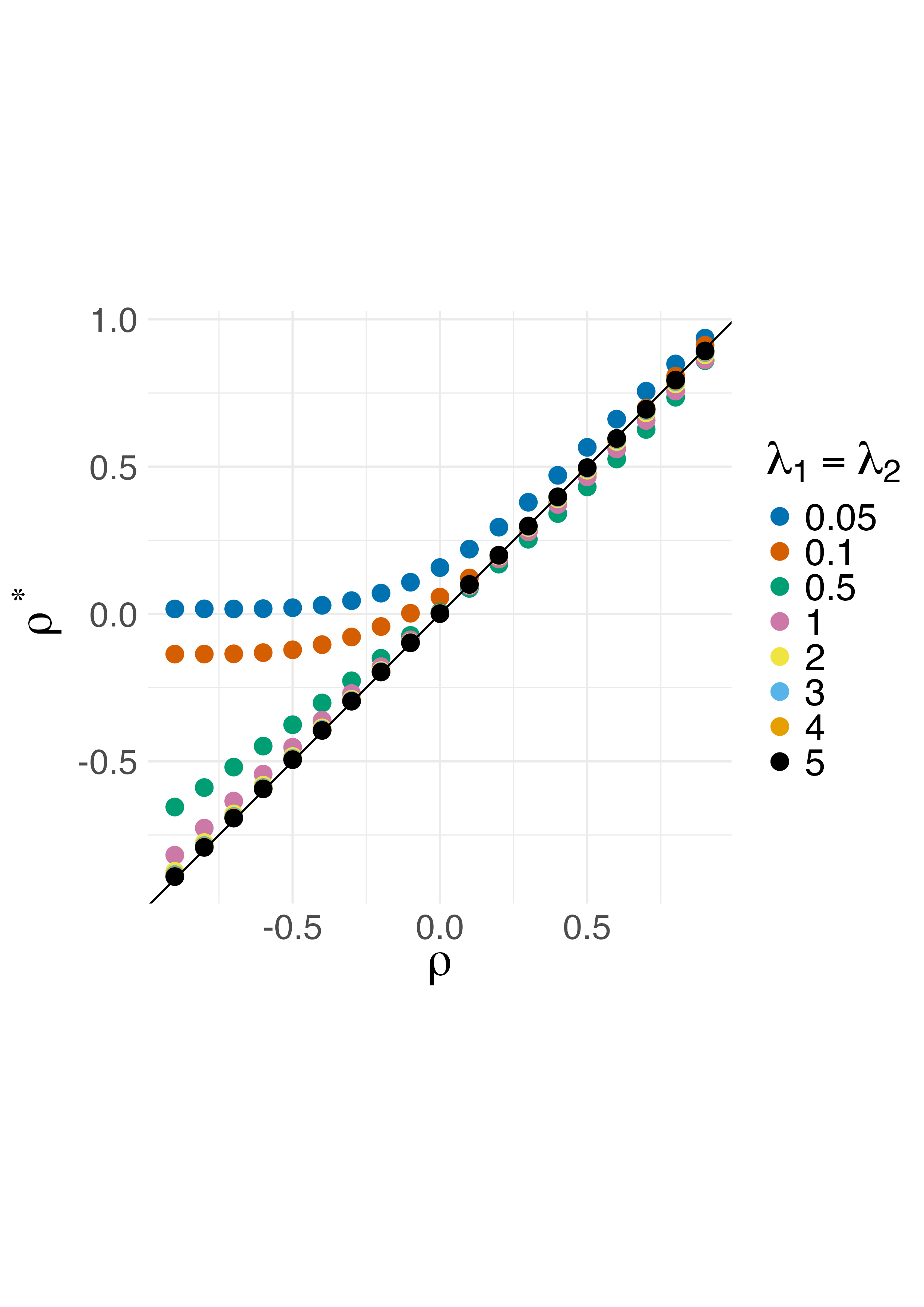}
    \caption{The relationship between $\rs$ and $\rho$ for different values of $\lx, \ly$.}
    \label{fig:rho_star}
\end{figure}
        
\autoref{fig:rho_star} shows the relationship between $\rs$ and $\rho$ for different values of lambda. We see that $\rs$ converges rather quickly to $\rho$, i.e. already when $\lx = \ly = 5$. However, the exact mapping $\psi_{Sk}$ is unknown.

\clearpage
\subsection{Additional details on Option 3}
\label{app:option3}

As a continuous benchmark, for a Gaussian copula with correlation $\rho$ and continous margins, it is know that $\tau_G(\rho) = \frac{2}{\pi} \arcsin (\rho)$. To build intuition, set $\lambda = \lx = \ly$ for visualization, compute $\tau_{Po}(\rho; \lambda, \lambda)$ from \autoref{eq:nik-tau} and plot it against a grid of values of $\rho \in [-1,1]$ in \autoref{fig:rhovstau}. The subscript $\tau_{Po}$ emphasizes the discrete (Poisson) margins and $\tau_G$ the continuous margins.

One can observe that the discrepancy $\Delta_{\lambda} (\rho):= \tau_{Po}(\rho; \lambda, \lambda) - \tau_G(\rho)$ varies monotonically with the magnitude of $\lambda$ on one hand while being more extreme in approaching the boundary cases (strong negative or positive dependence). Secondly, $\Delta_{\lambda} (\rho)$ is already small for $\lambda \approx 10$, making a Gaussian-limit approximation plausible. This is consistent with \citet{nikoloulopoulosCopulaBasedModelsMultivariate2013}. However, as our focus is the low-count regime, these differences $\Delta_{\lambda} (\rho)$ are non-negligeble and will matter in the simulations.

\begin{figure*}[t!]
    \centering
    \includegraphics[width=\textwidth]{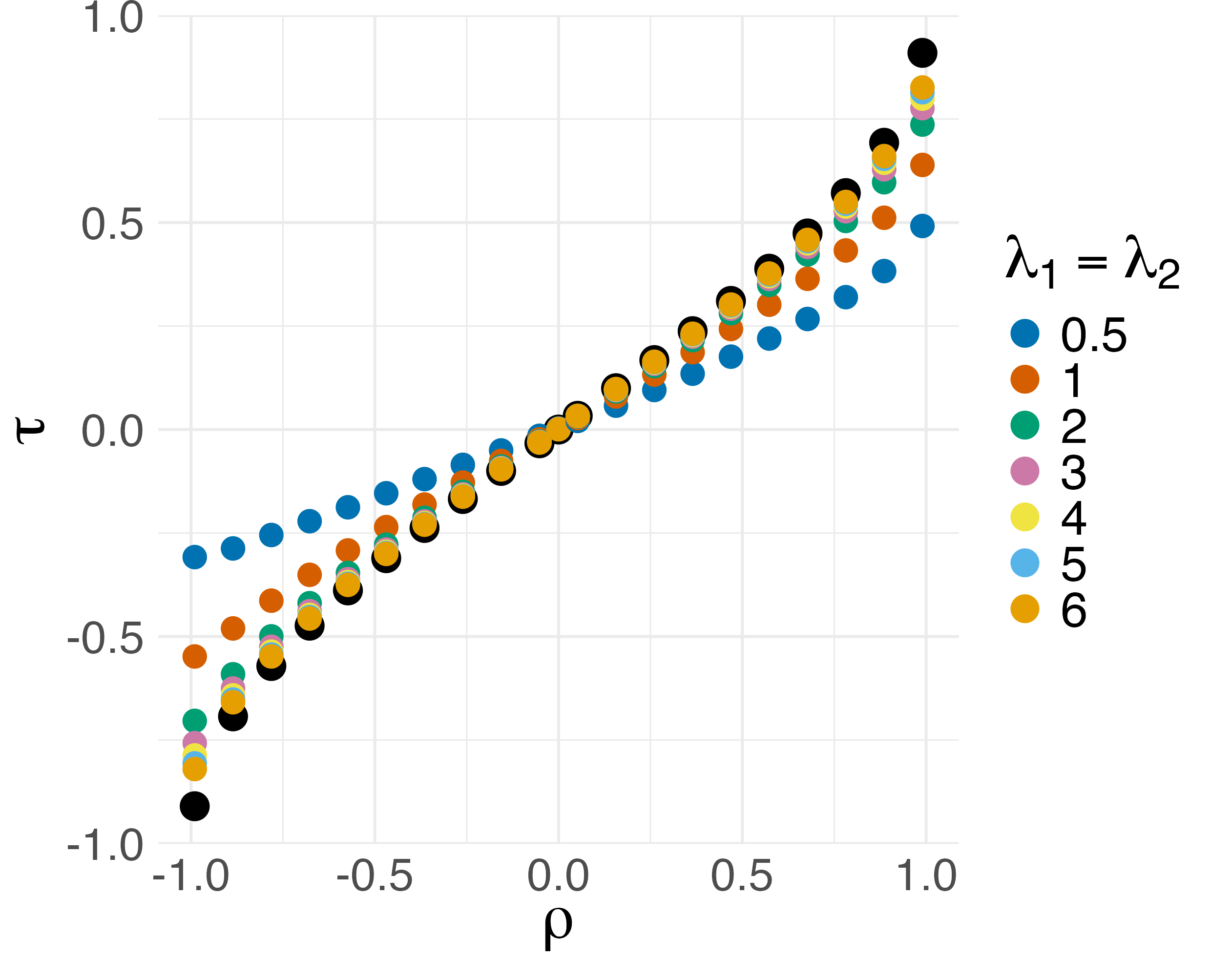}
    \caption{The relationship between $\frac{2}{ \pi} \arcsin (\rho)$ and corrected Kendall's tau.}
    \label{fig:rhovstau}
\end{figure*}

\clearpage
\subsection{Proof of \autoref{prop:tau-asin}}
\label{app:prop-tau-asin}

\begin{proof}[Proof of \autoref{prop:tau-asin}]
    Let 
    $$
    m(\rho, \lx, \ly) = a(\lx, \ly) \left( 2/ \pi \right) \arcsin \left(  b(\lx, \ly) \rho \right) + c(\lx, \ly).
    $$ 

    As $m(\rho, \lx, \ly)$ is an approximation to Kendall's tau, it must satisfy the same properties as the Kendall's tau itself, albeit corrected for discreteness (see Corollary 5.1.2 in \citet{nelsenIntroductionCopulas2006}), namely

    \begin{enumerate}[(i)]
        \item \label{p:1} $\lim_{\lx, \ly \to \infty} m(\rho, \lx, \ly) = \left( 2/ \pi \right) \arcsin \left(   \rho \right) \in [-1,1]$, 
        \item \label{p:2} $m(0, \lx, \ly) = 0$ $\forall \lx, \ly$,
        \item \label{p:3} $\lim_{\lx, \ly \to \infty} m(-\rho, \lx, \ly) = m(\rho, \lx, \ly)$,  
        \item \label{p:4} $m(\rho, \lx, \ly)$ is nondecreasing in rho $\forall \lx, \ly$.
    \end{enumerate}

    To satisfy the properties (\ref{p:1})-(\ref{p:4}) above, one must impose $a(\lx, \ly) \in (0,1]$ to satisfy properties (\ref{p:1}), (\ref{p:3}) and (\ref{p:4}). Furthermore, the domain of the arcsin function is $[-1,1]$, and since $\rho \in [-1,1]$ by construction, we must have $b(\lx, \ly) \in (0,1]$. Next, to fulfill property (\ref{p:2}) we must have that $m(0, \lx, \ly) = 0$, therefore imposing $c(\lx, \ly) = 0$. Hence, $m(\rho, \lx, \ly)$ reduces to 

    $$
    m(\rho, \lx, \ly) = a(\lx, \ly) \left( 2/ \pi \right) \arcsin \left(  b(\lx, \ly) \rho \right).
    $$ 

    Now, 
    
    \begin{align*}
        & \lim_{\lambda_j \to \infty} \posq{j} = \lim_{\lambda_j \to \infty} e^{-2 \lambda_j} \frac{e^{2 \lambda_j}}{ 2\sqrt{ \pi \lambda_j}} \left(1 + \frac{1}{16 \lambda_j} + \frac{9}{258 \lambda_j^2} + \cdots \right) = 0, \quad \text{(1)} \\
        & \lim_{\lx, \ly \to \infty} \pr{X_1=X_2, Y_1=Y_2} = \lim_{\lx, \ly \to \infty} \sum_{x=0}^{\infty} \sum_{y=0}^{\infty} h^2(x,y) = 0,
    \end{align*}

    where $(1)$ holds due to asymptotic expansion for large arguments of modified Bessel functions of first kind (equation $(9.7.1)$ in \citet{AbramowitzHandbookMathematical}).

    Hence, we must have that

    \begin{align*}
        &\lim_{\lx, \ly \to \infty} a(\lx, \ly) = \lim_{\lx, \ly \to \infty} b(\lx, \ly) = 1.
    \end{align*}
\end{proof}

\subsection{Gradient of the log-likelihood}
\label{app:grad}

\subsubsection{Gradient with respect to unconstrained correlation parameters}
\label{app:grad-cor}

In this section, we compute the log-likelihood gradient with respect to the set of unconstrained parameters $\bz$. It is convenient to work with Jacobian matrices. First, let $A \in \R^{d \times d}$ be a symmetric matrix, $\vect{A}$ denote the usual vec operator, $\vech{\cdot}$ the $d(d+1)/2$ vector that is obtained from $\vect{A}$ by eliminating all supradiagonal elemenents and $\vechs{\cdot}$ the strict half-vectorization operator that would only stack lower-triangular entries of $A$ without the main diagonal. Let $\bp = \lrep \lrep^{\top}$ with $\lrep$ lower triangular. Furthermore, let $\boldsymbol{K}_d$ be the $d^2 \times d^2$ commutation matrix, so that $\vect{\lrep^{\top}}=\boldsymbol{K}_d \vect{\lrep}$, $D_d$ the $d^2 \times d(d+1)/2$ duplication matrix such that $D_d \vech{\bp} = \vect{\bp}$ and $M_p$ a $p \times d(d+1)/2$ selection matrix such that $\vechs{\bp} = M_p \vech{\bp}$. Define $D_d^+ := \left(D_d^{\top} D_d\right)^{-1} D_d^{\top}$, the Moore-Penrose inverse of $D_d$ \citep{MatrixDifferentialCalculus}. 

The commutation matrix $K_d$ is an orthogonal permutation matrix that allows reordering of $\vect{\lrep}$ to obtain $\vect{\lrep^{\top}}$. It allows to commute two matrices of a Kronecker product. The duplication matrix $D_d$ allows to eliminate all the upper-diagonal elements of a square symmetric matrix. Finally, selection matrix allows to select specific elements of a matrix. In this case, it eliminates the main diagonal entries of a square symmetric matrix.
    
Let

\begin{align*}
    \boldsymbol{s}_{P} \left( \bg; \bys\right) &:= \frac{\partial \ell \left( \bg ; \bys\right)}{\partial \vechs{\bp}}
\end{align*}

denote the score function of the original log-likelihood (\autoref{eq:ll}) with respect to the vectorized correlation matrix $\br = \vechs{\bp}$. Then, by the chain rule,

$$
\frac{\partial \ell \left( \bps ; \bys\right)}{\partial \bz }= \boldsymbol{s}_{P} \left( \bg; \bys\right)^{\top} \left[\frac{\partial \vechs{\bp}}{\partial \bo }\right]^{}  \left[\frac{\partial \bo }{\partial \bz }\right]^{}.
$$

First, we have

\begin{align*}
    \boldsymbol{s}_{P} \left( \bg; \bys\right) &= \frac{\partial \ell \left( \bg ; \bys\right)}{\partial \vechs{\bp}} = \sum_{r=1}^n \frac{1}{h\left(\by^{(r)} ; \bg\right)} \frac{\partial h\left(\by^{(r)} ; \bg\right)}{\partial \vechs{\bp}} = \\
    & = \sum_{r=1}^n \frac{1}{h\left(\by^{(r)} ; \bg\right)} \frac{\partial }{\partial \vechs{\bp}} \sum_{t_1=0}^{1} \cdots \sum_{t_d=0}^{1} (-1)^{t_1 + \ldots + t_d}  H \left( \left( \by^{(r)} - \btt\right); \bg\right).                  
\end{align*}

Set $\ba:=\left(\Phi^{-1}\left(F_{Y_1}\left(y_1 -t_1\right)\right), \ldots, \Phi^{-1}\left(F_{Y_d}\left(y_d-t_d\right)\right)\right)$, $\boldsymbol{t} \in \{0,1\}^d$,  $b_m := \Phi^{-1}\left(F_{Y_m}\left(y_m -t_m\right)\right)$, $m = 1, \ldots, d$ and $\ba_S$ is vector $\ba$ without entries $b_i,b_j$. By using reduction formula in \citet{plackettReductionFormulaNormal1954}, we obtain

$$
\frac{\partial}{\partial \rho_{i j}} \Phi_d(\ba ; \bp)= \varphi_2\left(b_i, b_j ; \rho_{i j}\right) \bbP\left(\boldsymbol{Z}_{S} \leq \ba_S \mid Z_i=b_i, Z_j=b_j\right) = \varphi_2\left(b_i, b_j ; \rho_{i j}\right) \Phi_{d-2}\big( \ba_S;\, \boldsymbol{\mu}_{S|ij},\,\boldsymbol{\Sigma}_{S \mid i j}\big),
$$

with $\varphi_2$ being the bivariate normal density function and $\Phi_{d-2}$ is MVN cdf with dimension $d-2$. Then, for any $d$, 

$$
\frac{\partial}{\partial \rho_{i j}} H(\mathbf{y}-\mathbf{t}; \bg)= \varphi_2\left(b_i, b_j ; \rho_{i j}\right) \Phi_{d-2}\big( \ba_S;\, \boldsymbol{\mu}_{S|ij},\,\boldsymbol{\Sigma}_{S \mid i j}\big),
$$

and by using \autoref{eq:rect-prob}, we finally have

\begin{equation}
        \label{eq:grad-rho}
        \begin{aligned}
            \frac{\partial}{\partial \rho_{i j}} h(\mathbf{y} ; \bg) & = \sum_{t_1=0}^{1} \cdots \sum_{t_d=0}^{1} (-1)^{t_1 + \ldots + t_d}  \varphi_2\left(b_i, b_j ; \rho_{i j}\right) \Phi_{d-2}\big( \ba_S;\, \boldsymbol{\mu}_{S|ij},\,\boldsymbol{\Sigma}_{S \mid i j}\big).
        \end{aligned}
\end{equation}

The detailed computation of the next two terms is provided in \citet{lucchettiSphericalParametrisationCorrelation2024}. Below is a (brief) summary of the important steps and the computation of the final gradient of the log-likelihood. For $\bp=\lrep \lrep^{\top}$ with $\lrep$ lower-triangular and parameterized by angles $\left\{\omega_{i 1}, \ldots, \omega_{i, i-1}\right\}$ on row $i= 2, \ldots, d$, we have that

\begin{align*}
    \frac{\partial \vect{\bp}}{\partial \bo } & = \frac{\partial \vect{\lrep \lrep^{\top}}}{\partial \bo } = \left(\boldsymbol{I}_{d^2} + \boldsymbol{K}_d\right) \left(\boldsymbol{I}_d \otimes \lrep\right) \frac{\partial \vect{\lrep^{\top}}}{\partial \bo },
\end{align*}

where $\boldsymbol{I}$ is the identity matrix, and $\otimes$ denotes the Kronecker product. From the relationship between $\vechs{\bp}, \vech{\bp}$ and $\vect{\bp}$, we have

\begin{align*}
    \frac{\partial \vechs{\bp}}{\partial \bo } & = M_p \frac{\partial \vech{\bp}}{\partial \bo } = M_p D_d^+ \frac{\partial \vect{\bp}}{\partial \boldsymbol{\omega}} = M_p D_d^+ \left(\boldsymbol{I}_{d^2} + \boldsymbol{K}_d\right) \left(\boldsymbol{I}_d \otimes \lrep\right) \frac{\partial \vect{\lrep^{\top}}}{\partial \bo }.
\end{align*}

For each row $\boldsymbol{l}_i$ of $\lrep$, only the $\boldsymbol{\alpha}_i$ part is nonzero. Index the angles row-wise as $\mathcal{K}= \left\{ (i,m): 1 \le m < i \le d\right\}$. The authors compute $\partial \boldsymbol{\alpha}_i / \partial \omega_{i m}$ row-by-row recursively which gives

\begin{align*}
    \frac{\partial q_{i j}}{\partial \omega_{i m}} &= \1_{m<j} \cos \omega_{i m} \prod_{k=1, k \neq m}^{j-1} \sin \omega_{i k} = \1_{m<j} \frac{\cos \omega_{i m}}{\sin \omega_{i m}}  q_{i j} = \1_{m<j} \cot \omega_{i m}  q_{i j},
\end{align*}

\begin{align*}
    \frac{\partial \alpha_{i j}}{\partial \omega_{i m}} &= \begin{cases}
        \1_{m<j} \left( \frac{\partial q_{i j}}{\partial \omega_{i m}} \right) \cos \omega_{i j}   - \1_{m=j}q_{i (j+1)}, \quad j<i, \\
        \frac{\partial q_{i i}}{\partial \omega_{i m}} = q_{i i} \cot \omega_{i m}, \quad j=i.
    \end{cases} 
\end{align*}

For each $(i,m) \in \mathcal{K}$, collect these into $\mathbf{g}_{i m} \in \mathbb{R}^d$, whose first $i$ entries are $\partial \alpha_{i 1} / \partial \omega_{i m}, \ldots, \partial \alpha_{i i} / \partial \omega_{i m}$ and zeros afterward such that:

If $j = 1,\ldots, i$:

\begin{align*}
    \left(g_{i m}\right)_j &= \frac{\partial \alpha_{i j}}{\partial \omega_{i m}} = \begin{cases}
    -q_{i 2}, \quad & j=1,m=1,\\
    \1_{m<j} q_{i j}  \cot \omega_{i m} \cos \omega_{i j}   - \1_{m=j}q_{i (j+1)}, \quad & 2 \leq j \leq i-1, \\
    q_{i i} \cot \omega_{i m}, \quad & j=i,
\end{cases} 
\end{align*}

and $\left(g_{i m}\right)_j = 0$ if $j>i$. Since $\vect{\lrep^{\top}}$ stacks the rows of $\lrep$, the angle $\omega_{i m}$ only affects the block corresponding to row $i$ :

$$
\frac{\partial \vect{\lrep^{\top}}}{\partial \omega_{i m}}= \mathbf{e}_i \otimes \mathbf{g}_{i m} \in \R^{d^2},
$$

with $\mathbf{e}_i$ the $i$-th canonical vector in $\mathbb{R}^d$ and 

$$
\frac{\partial \vect{\lrep^{\top}}}{\partial \boldsymbol{\omega}} = \begin{bmatrix}
    \mathbf{e}_i \otimes \mathbf{g}_{i m}
\end{bmatrix}_{(i,m)\in \mathcal{K}}.
$$

Finally,

\begin{align*}
    \frac{\partial \omega_{i,j}}{\partial \zeta_{i j}} & = \pi \Lambda (\zeta_{i j}) \left\{1 - \Lambda (\zeta_{i j})\right\}, \quad \Lambda(x) = \frac{\exp(x)}{1 + \exp(x)}.
\end{align*}

Putting the pieces together, the scalar gradient component for a fixed $(i,m) \in \mathcal{K}$ is

\begin{equation}
    \label{eq:grad-zeta-final}
    \frac{\partial \ell \left( \bg ; \bys\right)}{\partial \zeta_{i m}} = \pi \Lambda\left(\zeta_{i m}\right)\left\{1-\Lambda\left(\zeta_{i m}\right)\right\} \boldsymbol{s}_{\bp} \left( \bg; \bys\right)^{\top} \left[M_p D_d^+ \left(\boldsymbol{I}_{d^2} + \boldsymbol{K}_d\right) \left(\boldsymbol{I}_d \otimes \lrep\right) \left[ \mathbf{e}_i \otimes \mathbf{g}_{i m}\right] \right] .
\end{equation}

\subsubsection{Gradient with respect to unconstrained Poisson means}

In this section we compute the gradient of the log-likelihood with respect to the rest of unconstrained parameters $\be$. For the Poisson means $\lambda_k = \exp (\eta_k)$, let

\begin{align*}
    \boldsymbol{s}_{\lambda_k} \left( \bg; \bys\right) &:= \frac{\partial \ell \left( \bg ; \bys\right)}{\partial \lambda_k}
\end{align*}

denote the score function of the original log-likelihood (\autoref{eq:ll}) with respect to $\lambda_k$. Then, by the chain rule,

$$
\frac{\partial \ell \left( \bps ; \bys\right)}{\partial \eta_k} = \frac{\partial \ell \left( \bg ; \bys\right)}{\partial \lambda_k} \frac{\partial \lambda_k}{\partial \eta_k}= \boldsymbol{s}_{\lambda_k} \left( \bg; \bys\right) \lambda_k.
$$

The first term is

\begin{align*}
    \boldsymbol{s}_{\lambda_k} \left( \bg; \bys\right) &= \frac{\partial \ell \left( \bg ; \bys\right)}{\partial \lambda_k}= \sum_{r=1}^n \frac{1}{h\left(\by^{(r)} ; \bg\right)} \frac{\partial h\left(\by^{(r)} ; \bg\right)}{\partial \lambda_k} = \\
    & = \sum_{r=1}^n \frac{1}{h\left(\by^{(r)} ; \bg\right)} \frac{\partial }{\partial \lambda_k} \sum_{t_1=0}^{1} \cdots \sum_{t_d=0}^{1} (-1)^{t_1 + \ldots + t_d}  H \left( \left(\by^{(r)} - \btt \right); \bg\right).
\end{align*}

By using reduction formula in \citet{plackettReductionFormulaNormal1954}, for a fixed corner only $b_k$-th entry of $\ba$ depends on $\lambda_k$, hence

$$
\frac{\partial}{\partial \lambda_{k}} H(\mathbf{y}^{(r)}-\mathbf{t}; \bg)
= \frac{\partial \Phi_d(\ba;\bp)}{\partial b_k}\cdot \frac{\partial b_k}{\partial \lambda_k}.
$$

With $\mathcal{T} =\{1,\dots,d\}\setminus\{k\}$, partition $\bp$ as $\bp_{\mathcal{T},k}$, with rows consisting of $\mathcal{T}$, column consisting of $k$, and $\bp_{\mathcal{T}\mathcal{T}}$, with rows and columns consisting of $\mathcal{T}$. Then the conditional MVN property \citep{tongMultivariateNormalDistribution2012} and \citet{plackettReductionFormulaNormal1954} yields

$$
\frac{\partial  \Phi_d( \ba; \bp)}{\partial b_k}
= \varphi(b_k)\; \Phi_{d-1}\big( \ba_{\mathcal{T}};\,\boldsymbol{\mu}_{\mathcal{T} \mid k},\,\boldsymbol{\Sigma}_{\mathcal{T} \mid k}\big),
$$

where $\boldsymbol{\mu}_{\mathcal{T} \mid k}=\bp_{\mathcal{T},k}b_k$ and $\boldsymbol{\Sigma}_{\mathcal{T} \mid k}=\bp_{\mathcal{T}\mathcal{T}}-\bp_{\mathcal{T},k}\bp_{\mathcal{T},k}^{\top}$. In addition, by the inverse function theorem,

$$
\frac{\partial b_k}{\partial \lambda_k}
= \frac{1}{\varphi(b_k)}\cdot \frac{\partial F_{Y_k}(y_k^{(r)} - t_k)}{\partial \lambda_k} = - \frac{f_{Y_k}(y_k^{(r)} - t_k)}{\varphi(b_k)},
$$

where $f_{Y_k}$ is the Poisson pmf. Since the joint pmf is given by inclusion-exclusion, and with the convention $F_{Y_j}(-1)=0$ (thus $b_j(-1)=-\infty$ for all $j= 1,\ldots,d$), and since $\boldsymbol{t}$ is discrete and the sum is finite, this finally gives

\begin{equation}
    \label{eq:grad-tau-final}
    \boldsymbol{s}_{\lambda_k} \left( \bg; \bys\right) = \sum_{r=1}^n \frac{1}{h\left(\by^{(r)} ; \bg\right)} \sum_{t_1=0}^{1} \cdots \sum_{t_d=0}^{1} (-1)^{t_1 + \ldots + t_d}  \left( - f_{Y_k}(y_k^{(r)} - t_k) \right) \; \PhiDm\!\Big(\ba_{\mathcal{T}}^{(r)};\,\boldsymbol{\mu}_{\mathcal{T} \mid k}^{(r)},\, \boldsymbol{\Sigma}_{\mathcal{T} \mid k}\Big).
\end{equation}

\section{Additional simulation results}
\label{app:sim-res}

\subsection{Simulation tables}
\label{app:sim-tab}

In this section, we add the additional tables for the simulation results for each setting and combination of parameters.

\begin{table}[h]
    \centering
    \begin{tabular}{|c|c|c|c|c|c|c|c|}
        \hline 
             Setting & Lambdas & Gradient &  Start &  Method &      $\hat{\lambda}_1$ &  $\hat{\lambda}_2$ &    $\hat{\rho}$ \\ \hline 
            $2.1$ & $2, 3$ & Analytic & Corr & Reparam & $0.063$ & $0.069$ & $0.041$ \\
            $2.2$ & $2, 3$ & Analytic & Option 1 & Reparam & $0.062$ & $0.066$ & $0.037$ \\
            $2.3$ & $2, 3$ & Numeric & Corr & Reparam & $0.060$ & $0.066$ & $0.037$ \\
            $2.4$ & $2, 3$ & Numeric & Corr & fitMvdc & $0.210$ & $0.324$ & $0.057$ \\
            $2.5$ & $2, 3$ & Numeric & Option 1 & Reparam & $0.060$ & $0.066$ & $0.037$ \\
            $2.6$ & $2, 3$ & Numeric & Option 1 & fitMvdc & $0.210$ & $0.324$ & $0.057$ \\
            $2.7$ & $0.5, 1$ & Analytic & Corr & Reparam & $0.041$ & $0.048$ & $0.054$ \\
            $2.8$ & $0.5, 1$ & Analytic & Option 1 & Reparam & $0.034$ & $0.037$ & $0.050$ \\
            $2.9$ & $0.5, 1$ & Numeric & Corr & Reparam & $0.035$ & $0.037$ & $0.050$ \\
            $2.10$ & $0.5, 1$ & Numeric & Corr & fitMvdc & $0.147$ & $0.278$ & $0.851$ \\
            $2.11$ & $0.5, 1$ & Numeric & Option 1 & Reparam & $0.035$ & $0.037$ & $0.050$ \\
            $2.12$ & $0.5, 1$ & Numeric & Option 1 & fitMvdc & $0.147$ & $0.278$ & $0.851$ \\
        \hline
    \end{tabular}
    \caption{RMSEs for $d=2$.}
    \label{tab:sim-res-d2}
\end{table}

\begin{table}[h]
    \centering
    \begin{tabular}{|c|c|c|c|c|c|c|}
        \hline 
        Setting & Gradient &  Start &  $\hat{\lambda}_1$ & $\hat{\lambda}_2$ & $\hat{\lambda}_3$ & $\hat{\lambda}_4$ \\ \hline 
        $4.1$ & Analytic & Corr & $0.039$ & $0.058$ & $0.090$ & $0.040$ \\
        $4.2$ & Analytic & Option 1 & $0.038$ & $0.055$ & $0.084$ & $0.040$ \\
        $4.3$ & Numeric & Option 1 & $0.039$ & $0.079$ & $0.144$ & $0.045$ \\
        \hline
    \end{tabular}
    \caption{RMSEs for $\hat{\boldsymbol{\lambda}}$, $d=4$.}
    \label{tab:sim-res-d4-lam}
\end{table}

\begin{table}[h]
    \centering
    \begin{tabular}{|c|c|c|c|c|c|c|c|c|}
        \hline 
            Setting & Gradient &  Start & $\hat{\rho}_{21}$ & $\hat{\rho}_{31}$ & $\hat{\rho}_{41}$ & $\hat{\rho}_{32}$ & $\hat{\rho}_{42}$ & $\hat{\rho}_{43}$ \\ \hline 
            $4.1$ & Analytic & Corr & $0.055$ & $0.050$ & $0.052$ & $0.049$ & $0.060$ & $0.050$ \\
            $4.2$ & Analytic & Option 1 & $0.045$ & $0.046$ & $0.042$ & $0.057$ & $0.036$ & $0.054$ \\
            $4.3$ & Numeric & Option 1 & $0.051$ & $0.073$ & $0.036$ & $0.069$ & $0.042$ & $0.075$ \\
        \hline
    \end{tabular}
    \caption{RMSEs for $\hat{\boldsymbol{\rho}}$, $d=4$.}
    \label{tab:sim-res-d4-rho}
\end{table}

\clearpage
\subsection{Simulation results plots}
\label{app:sim-res-plots}

Below are the plots for the simulation results under each setting. 

\subsubsection{Lambdas}

In this subsection, the results for the simulations are shown for the Poisson parameters $\lambda$ under each setting.

\begin{figure}[ht]
    \centering
    \includegraphics[width=0.75\textwidth]{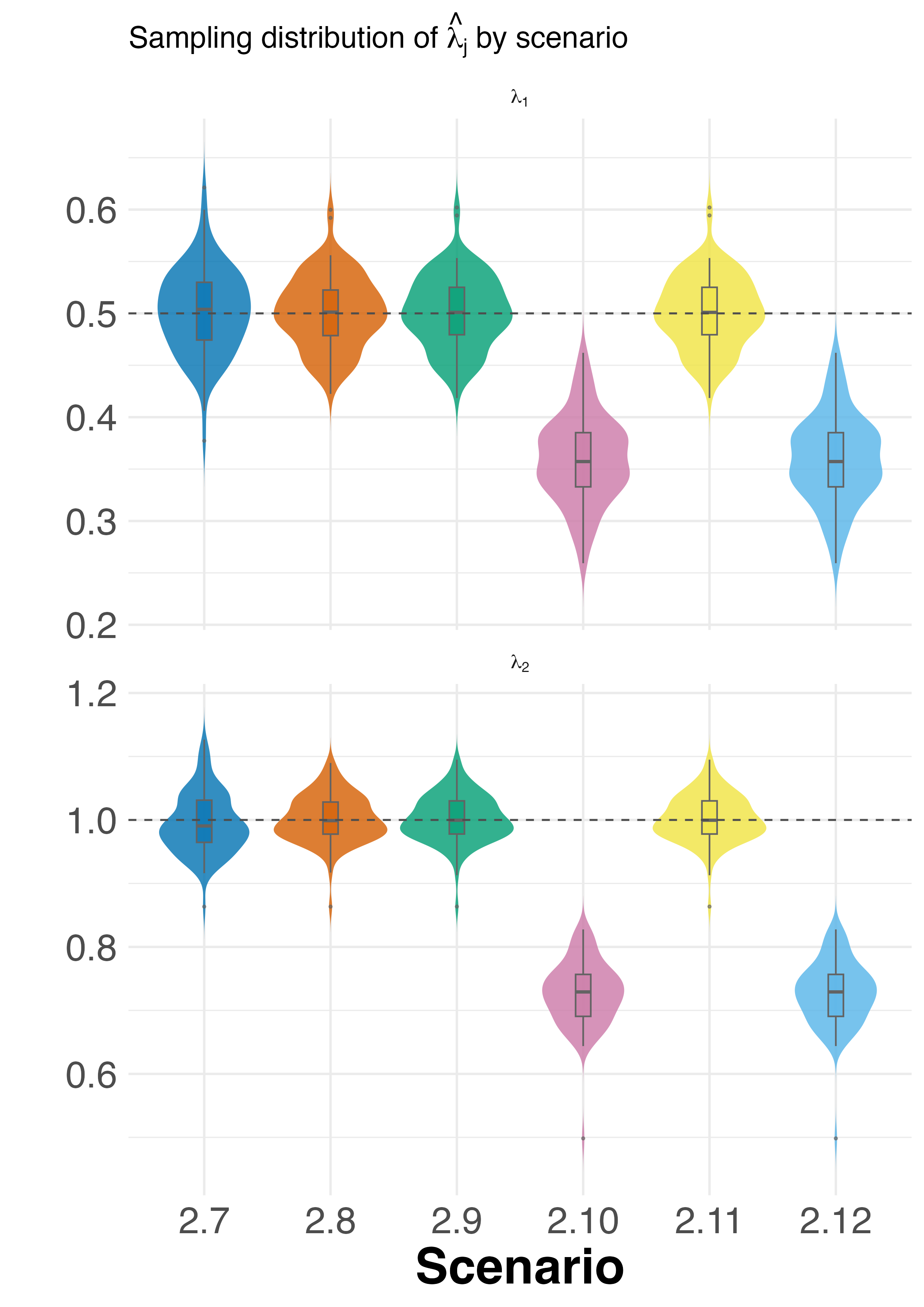}
    \caption{Violin plots for $\boldsymbol{\lambda} = (0.5,1)$ for different scenarios.}
    \label{fig:violin-lam-d2-small}
\end{figure}

\begin{figure}[ht]
    \centering
    \includegraphics[width=0.75\textwidth]{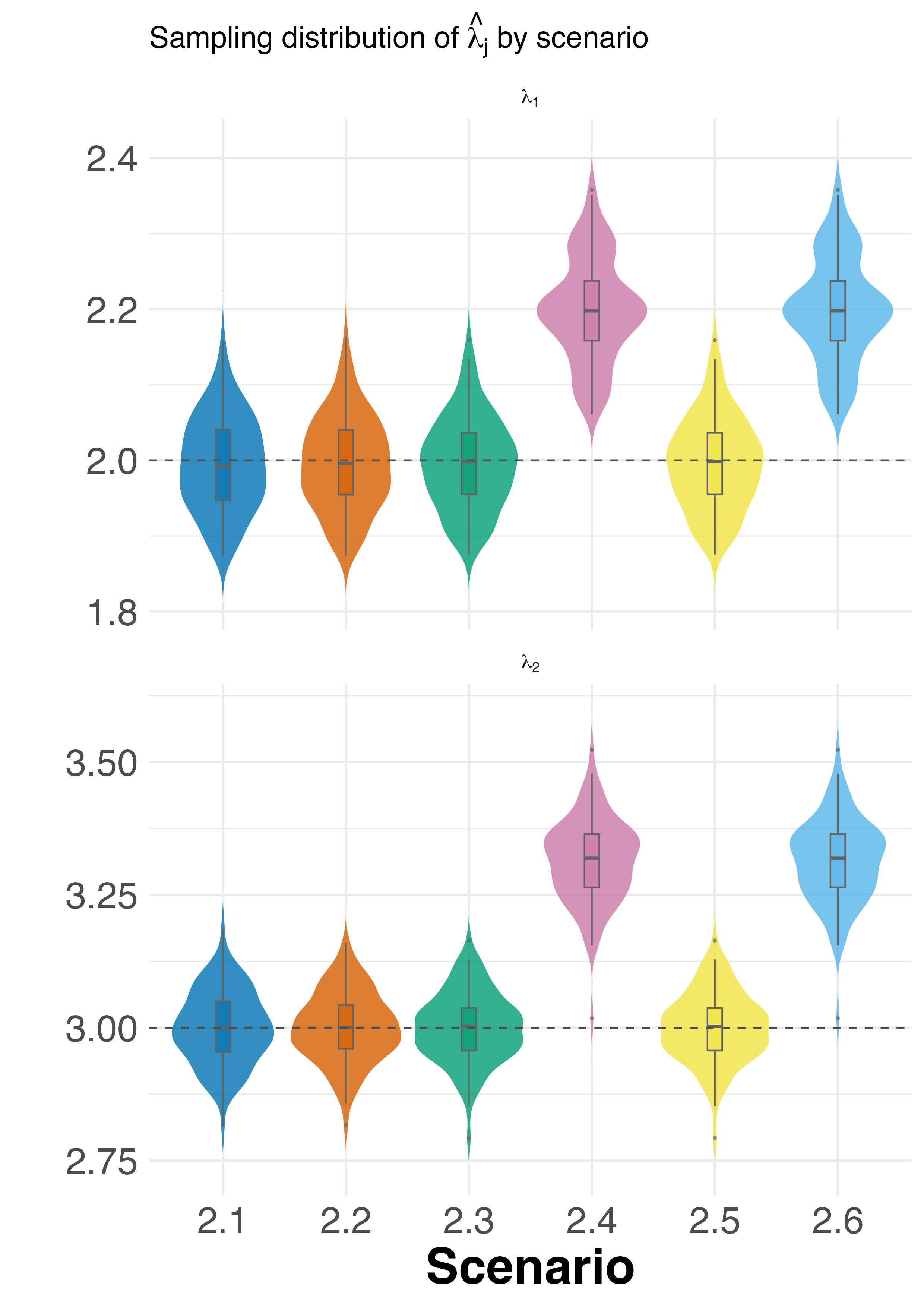}
    \caption{Violin plots for $\boldsymbol{\lambda} = (2,3)$ for different scenarios.}
    \label{fig:violin-lam-d2-moderate}
\end{figure}

\begin{figure}[ht]
    \centering
    \includegraphics[width=0.75\textwidth]{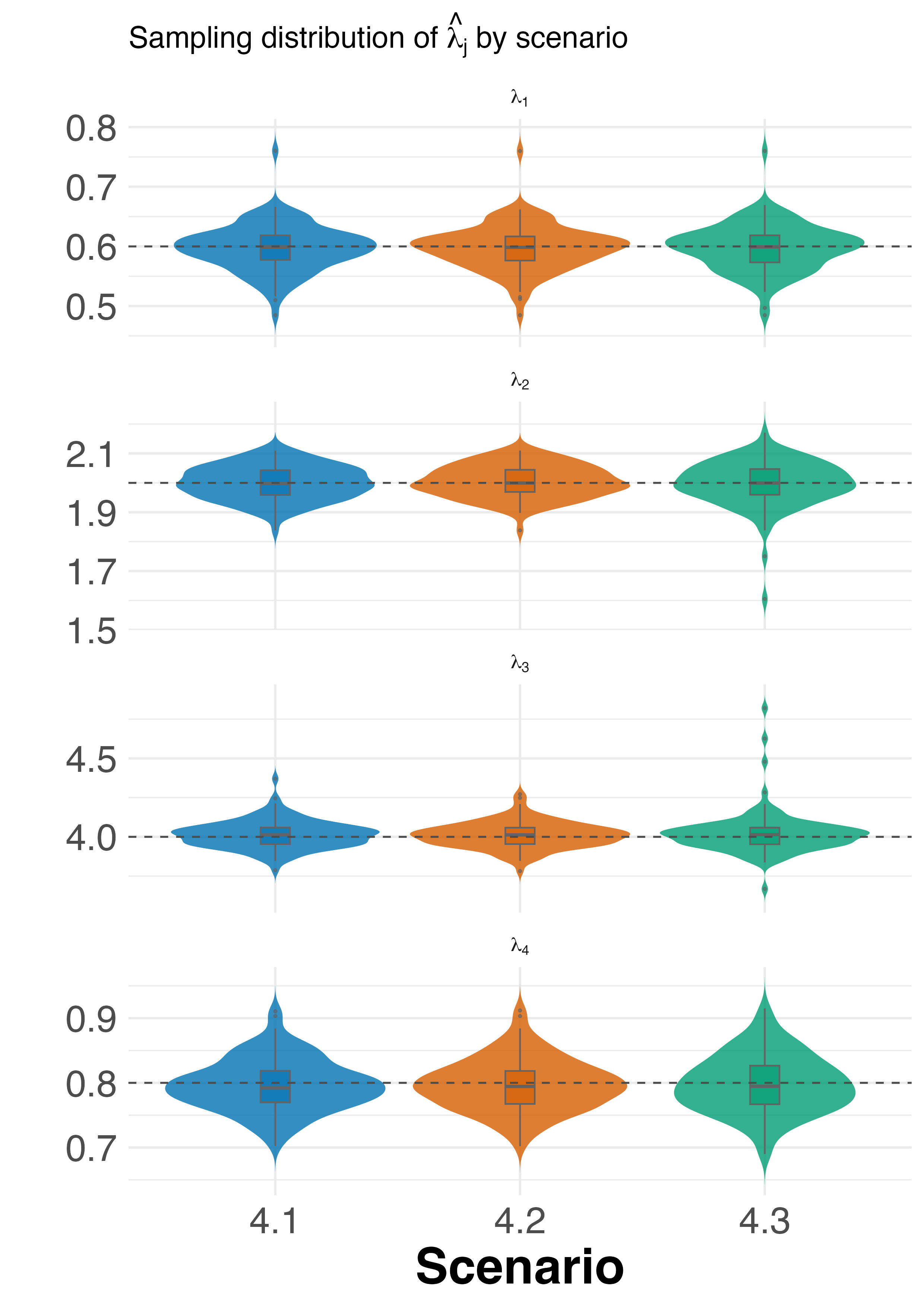}
    \caption{Violin plots for $\boldsymbol{\lambda}, d = 4$ for different scenarios.}
    \label{fig:violin-lam-d4}
\end{figure}

\clearpage
\subsubsection{Rhos}

In this subsection, the results for the simulations are shown for correlation parameters $\rho$ under each setting.

\begin{figure}[ht]
    \centering
    \includegraphics[width=0.75\textwidth]{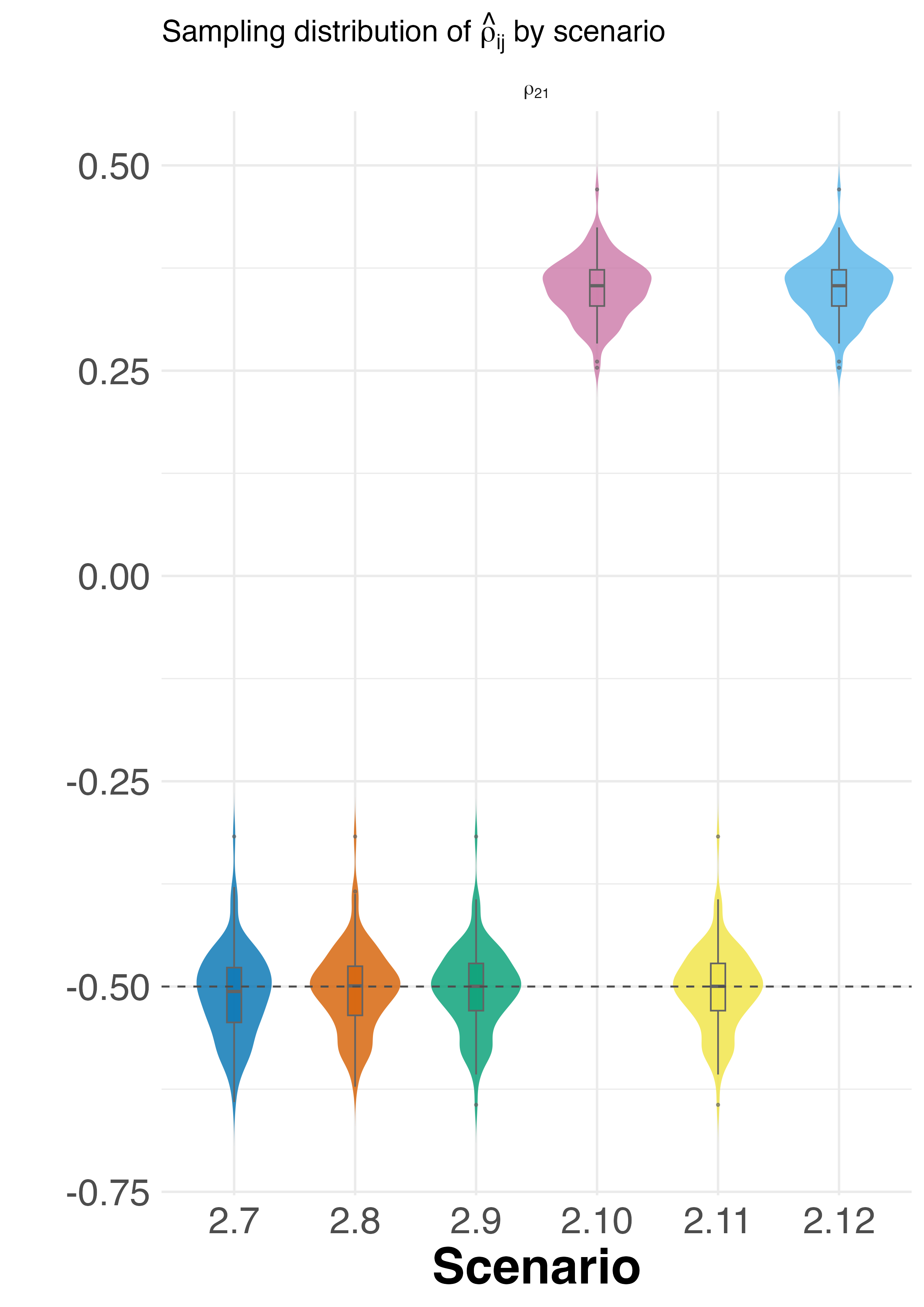}
    \caption{Violin plots for $\rho_{21}$ when $\boldsymbol{\lambda} = (0.5,1)$ for different scenarios.}
    \label{fig:violin-rho-d2-small}
\end{figure}

\begin{figure}[ht]
    \centering
    \includegraphics[width=0.75\textwidth]{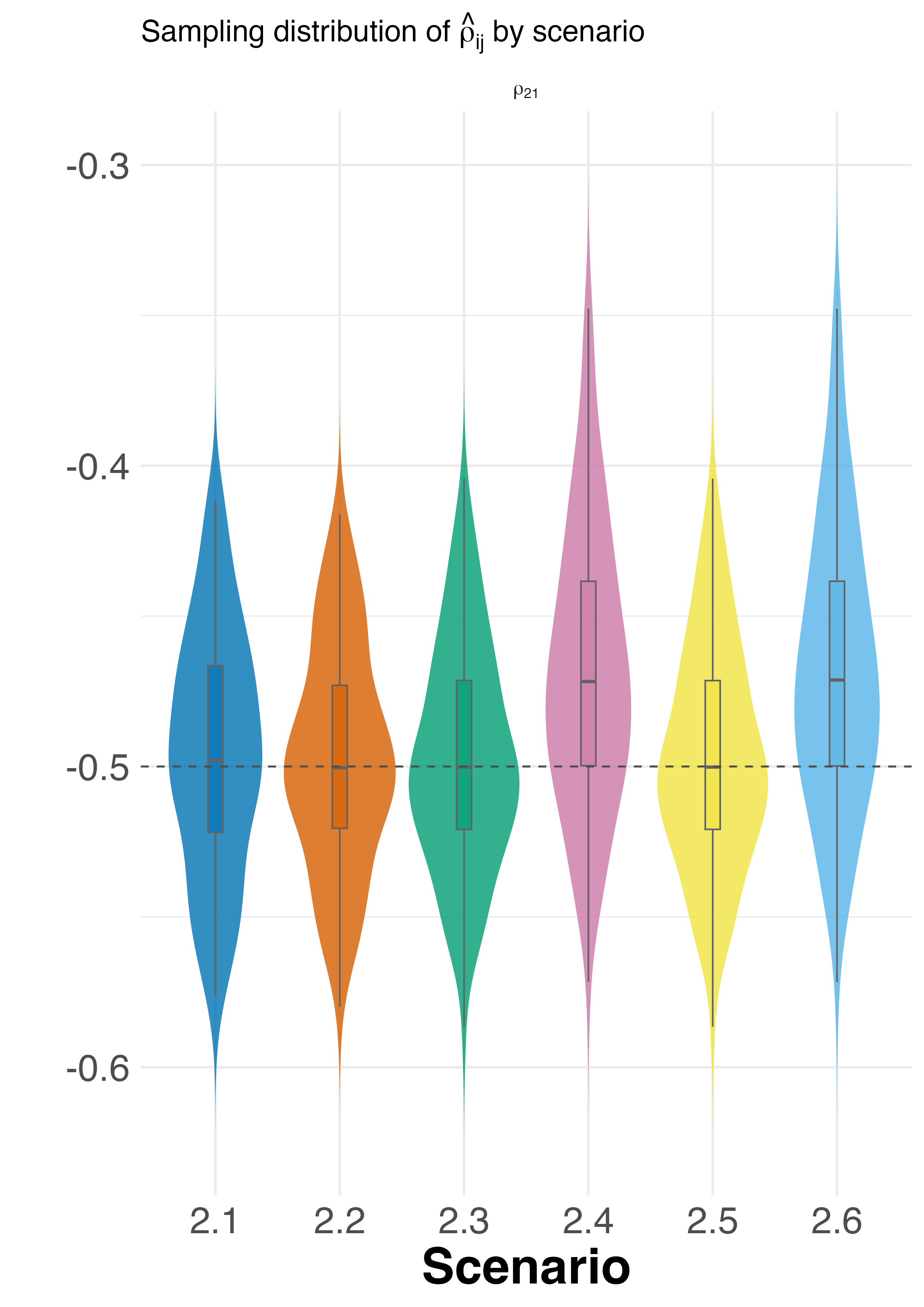}
    \caption{Violin plots for $\rho_{21}$ when $\boldsymbol{\lambda} = (2,3)$ for different scenarios.}
    \label{fig:violin-rho-d2-moderate}
\end{figure}

\begin{figure}[ht]
    \centering
    \includegraphics[width=0.75\textwidth]{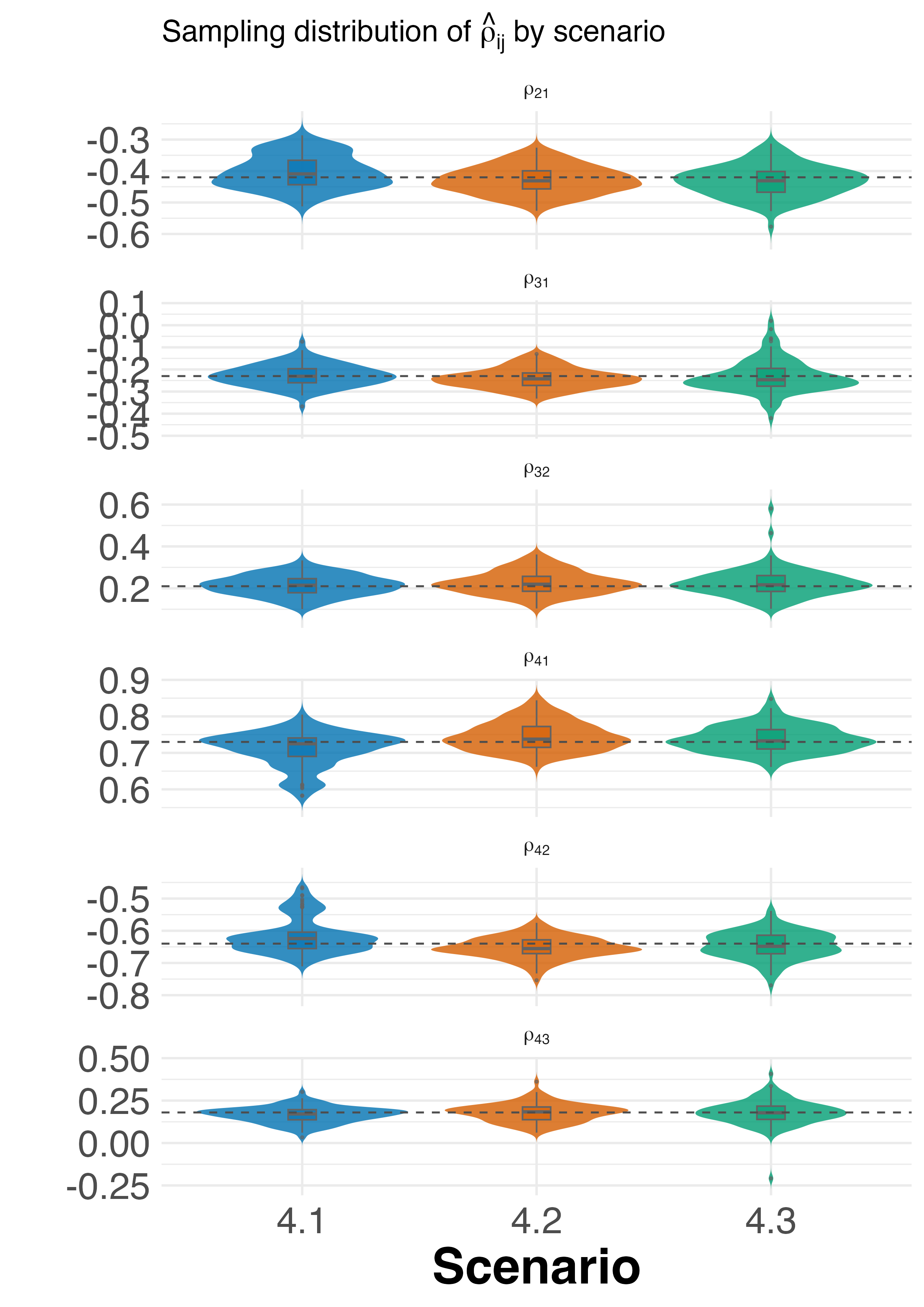}
    \caption{Violin plots for $\br, d = 4$ for different scenarios.}
    \label{fig:violin-rho-d4}
\end{figure}

\end{appendices}

\clearpage 
\newpage
\printbibliography
\end{document}